\documentclass[a4paper,12pt]{amsart}

\usepackage{amsmath,amssymb,amsthm}
\usepackage{graphicx,stmaryrd,xfrac,enumitem}
\usepackage{hyperref}

\newcommand{\CC}{\mathbb{C}}
\newcommand{\EE}{\mathbb{E}}

\newcommand{\RR}{\mathbb{R}}

\newcommand{\ZZ}{\mathbb{Z}}

\newcommand{\cA}{\mathcal{A}}

\newcommand{\cH}{\mathcal{H}}

\newcommand{\cW}{\mathcal{W}}

\newcommand{\cZ}{\mathcal{Z}}

\renewcommand{\a}{\alpha}
\newcommand{\D}{\Delta} 
\renewcommand{\d}{\delta} 
\newcommand{\G}{\Gamma}

\newcommand{\g}{\gamma}

\renewcommand{\b}{\beta} 
 
\newcommand{\Om}{\Omega}
 
\renewcommand{\S}{\Sigma} 
\newcommand{\s}{\sigma}

\newcommand{\eps}{\varepsilon}

\newcommand{\el}{\langle} 
\newcommand{\er}{\rangle}
\newcommand{\tr}{\mathrm{tr}}

\newcommand{\sd}{\triangle}

\newcommand{\fk}{{\sc fk}}
\newcommand{\tfim}{{\sc tfim}}

\newcommand{\lra}{\leftrightarrow}

\renewcommand{\b}{\beta}
\newcommand{\oo}{\infty}

\newcommand{\sm}{\setminus}

\newcommand{\es}{\varnothing}
\newcommand{\se}{\subseteq}

\newcommand{\ol}{\overline}

\newcommand{\one}{\hbox{\rm 1\kern-.27em I}}

\newcommand{\w}{\circ}
\renewcommand{\b}{\bullet}
\newcommand{\up}{\uparrow}
\newcommand{\dn}{\downarrow}
\newcommand{\lt}{\leftarrow}
\newcommand{\rt}{\rightarrow}
\newcommand{\Dom}{\Om}

\renewcommand{\Re}{\mathrm{Re}}
\renewcommand{\Im}{\mathrm{Im}}
\newcommand{\rv}{\mathrm{v}}
\newcommand{\rh}{\mathrm{h}}
\newcommand{\itr}{\mathrm{int}}

\newcommand{\be}{\begin{equation}}
\newcommand{\ee}{\end{equation}}

\def\upu{\lower 2mm \hbox{\includegraphics{fermionic.32}}}
\def\dnu{\lower 2mm \hbox{\includegraphics{fermionic.33}}}

\def\rtw{\hbox{\includegraphics{fermionic.34}}}
\def\ltw{\hbox{\includegraphics{fermionic.35}}}
\def\uprtw{\lower 2mm \hbox{\includegraphics{fermionic.36}}}
\def\dnrtw{\lower 2mm \hbox{\includegraphics{fermionic.37}}}

\def\upv{\lower 2mm \hbox{\includegraphics{fermionic.38}}}
\def\dnv{\lower 2mm \hbox{\includegraphics{fermionic.39}}}
\def\upltw{\lower 2mm \hbox{\includegraphics{fermionic.40}}}
\def\dnltw{\lower 2mm \hbox{\includegraphics{fermionic.41}}}

\newtheoremstyle{slthm}
     {}
     {\baselineskip}
     {\slshape}
     {\parindent}
     {\scshape}
     {.}
     { }
     {}

\theoremstyle{slthm}
\newtheorem{definition}{Definition}[section]
\newtheorem{theorem}[definition]{Theorem}
\newtheorem{proposition}[definition]{Proposition}
\newtheorem{lemma}[definition]{Lemma}

\newtheorem{remark}[definition]{Remark}

\allowdisplaybreaks

\title
{Fermionic observables in the transverse Ising  chain}
\author{Jakob E. Bj\"ornberg}
\thanks{Department of Mathematics,
University of Gothenburg, Sweden,
e-mail: jakob.bjornberg@gmail.com}
\date{\today}

\begin{document}
\maketitle

\begin{abstract}
We introduce a notion of s-holomorphicity suitable for
certain quantum spin systems in one dimension, and define two
observables in the critical transverse-field Ising model which have this
property.  The observables are defined using graphical
representations in the complex plane, and are analogous to 
Smirnov's \fk--Ising 
and spin-Ising observables, respectively.  
We also briefly discuss scaling-limits of these
observables.
\end{abstract}

\section{Introduction}

Recent years have seen tremendous progress on the understanding of
planar models in statistical physics, particularly the (classical)
Ising model at criticality.  A major breakthrough in this area was 
the definition, and proof of convergence to conformally covariant
scaling limits,
of fermionic observables in the critical Ising
model, first on the square lattice by Smirnov~\cite{smi-towards,smi-fk},
and later on all isoradial graphs  
by Chelkak and Smirnov~\cite{ch-sm}.

The fermionic observables enjoy a crucial property 
called s-holo{-}morphicity, a strong form of discrete analyticity.  
Besides satisfying a discrete version of the Cauchy--Riemann
relations, if a function $F_\d$ 
is s-holomorphic then
one may define a discrete primitive 
$H_\d=\Im\big(\int^\d F_\d^2\big)$ of its
\emph{square}. 
Moreover this function $H_\d$ is very close to being
(discrete) harmonic. 
When combined with control of the behaviour of $H_\d$ at the boundary
of the domain, this allows to deduce convergence of the
fermionic observables from convergence of  solutions
to discrete boundary-value problems.

The identification of these and related observables and their scaling limits has
subsequently led to some outstanding results on the
critical planar Ising model, settling several predictions from
conformal field theory.  This includes convergence of 
the energy-density~\cite{en-dens},
correlation functions~\cite{correlations}, as well as
interfaces to SLE-curves~\cite{interfaces} and
loops to CLE-processes~\cite{ben-hong,ke-sm}, to mention but a few.
There has also been work on extending the definition of
s-holomorphicity to general graphs~\cite{cimasoni}.

In this note we start to consider similar questions in the context of
one-dimensional quantum spin-systems, specifically the
transverse-field (quantum) Ising model, hereafter abbreviated \tfim.
This model has  Hamiltonian given by
\be\label{ham_eq}
-\cH_N=J\sum_{x=1}^{N-1} \s_x^{(3)}\s_{x+1}^{(3)}
+h\sum_{x=1}^N \s_x^{(1)},
\mbox{ acting on } \otimes_{x=1}^N \CC^2,
\ee
where 
$\s^{(3)}=\big(\begin{smallmatrix}
1 & 0 \\
0 & -1
\end{smallmatrix}\big)$ and
$\s^{(1)}=\big(\begin{smallmatrix} 
0 & 1 \\
1 & 0
\end{smallmatrix}\big)$
are the spin-$\tfrac12$ Pauli matrices,
and $J,h>0$ give the coupling- and transverse-field-strengths,
respectively.  (For $h=0$ this is just the classical Ising model.)
We will be working with the ground-state (zero
temperature), where the model is known to undergo a phase-transition
as the ratio $h/J$ is varied, at the critical point $h/J=1$~\cite{pfeuty}.
The phase-transition is continuous~\cite{bjogr}. 

It is well-known that the {\tfim}
in $d$ dimensions possesses a graphical, probabilistic
representation in $\ZZ^d\times\RR$, and it behaves in many ways like a
classical Ising model in $d+1$ dimensions, see e.g.\ the results
in~\cite{B-irb,B-van}.   One may thus ask if
the results mentioned above, on conformal invariance in the
two-dimensional classical Ising model at criticality,
have analogs in the one-dimensional
quantum model?

This note is a first step in this direction.  We introduce a notion of
s-holomorphicity for functions on 
$\ZZ+i\RR\se\CC$;  we
show that functions that satisfy this enjoy (analogs of) the key
properties that hold in the classical case;
and we define two observables  in the critical
 {\tfim} which we show to be
s-holomorphic.   

The graphical representations 
that we consider may be obtained as
limits of classical counterparts on $\ZZ+i(\eps\ZZ)$ as
$\eps\to0$.  The latter graphs are all isoradial, and some of the key
quantities we work with can be interpreted as limits of the
corresponding quantities for isoradial graphs~\cite{ch-sm}.
We give examples of this 
in Section~\ref{isorad-sec}.  
However, for all our definitions and results we work
directly in the `continuous' setting $\ZZ+i\RR$ and the
rescaled version $\d\ZZ+i\RR$.

We do not go into the details for scaling limits (as $\d\to0$)
of our  observables here, but we expect this
to be very similar to the classical case.  As we discuss in
Section~\ref{disc-sec}, we expect
analogous reasoning and estimates to show
that our
observables converge to the \emph{same} scaling-limits as their
classical  counterparts.

\subsection*{Outline and main contributions}
After reviewing the graphical
representations of the {\tfim}  in Section~\ref{graph-sec}, we
 give our definition of s-holomorphicity in Section~\ref{s-hol-sec},
and prove some key properties of s-holomorphic functions in
Proposition~\ref{H-prop}.  
We introduce and study our two fermionic observables in
Sections~\ref{fk-obs-sec} and~\ref{spin-obs-sec}, respectively.  The
main results 
are that these observables satisfy our definition of s-holomorphicity,
stated precisely in Theorems~\ref{s-hol-thm-fk}
and~\ref{s-hol-thm-spin}. 

\subsection*{Bibliographical remark}
Shortly after this paper was made public, Li~\cite{li}
announced a complete proof of convergence of
the {\fk}-observable considered here, as well as the
{\fk}-interface to SLE$_{16/3}$, in the scaling limit.
Li independently arrived at equivalent definitions of the
{\fk}-observable and s-holomorphicity as presented here,
and supplied the details necessary to prove convergence.
He does not consider the spin-observable.  Most likely
his results are useful for proving convergence of that
observable as well.

\section{Graphical representations of the TFIM}
\label{graph-sec}

We briefly review three graphical representations of the {\tfim}.
They may be obtained using a Lie--Trotter expansion, see 
e.g.~\cite{akn,bjo_phd,bjogr,ioffe_geom} for details.
We also present a version of the Kramers--Wannier duality;  as for the
classical case, this allows us to easily identify the critical parameters
of the model (but for rigorous proofs
see~\cite{pfeuty,bjogr}).  

We write the partition function
$\cZ_{N,\beta}=\cZ_{N,\beta}(h,J)=\tr(e^{-\beta \cH_N})$ 
where $\cH_N$ is the Hamiltonian~\eqref{ham_eq}, and 
$\beta>0$ is the inverse-temperature.  For illustration we will also
consider the two-point correlation 
\[
\el\s^{(3)}_x\s^{(3)}_y\er_{N,\beta}=
\tr(\s^{(3)}_x\s^{(3)}_ye^{-\beta \cH_N})/\cZ_{N,\beta}.
\]
Thermodynamic limits are obtained for $N\to\oo$, and
the ground-state is obtained by also letting
$\beta\to\oo$.

The {\tfim} on $\{1,\dotsc,N\}$ 
maps onto stochastic models in the rectangular domain 
$\Dom=[1,N]+i[0,\beta]\se\CC$.  We write 
\be
\Dom^\b=\{1,\dotsc,N\}+i[0,\beta],\quad
\Dom^\w=(\sfrac12+\{1,\dotsc,N-1\})+i[0,\beta].
\ee
We will let $\xi^\b$ and $\xi^\w$ denote independent Poisson processes
on $\Dom^\b$ and $\Dom^\w$, respectively.  Their respective rates
will be denoted $r^\b$ and $r^\w$ and will be functions of $h$ and
$J$.  We 
write $\EE_{r^\b,r^\w}[\cdot]$ for the law (expectation operator)
governing them, and $\xi=\xi^\b\cup\xi^\w$.  Elements of $\xi^\b$ will
be represented graphically by $\times$ and called `cuts';  an element
$(x+\sfrac12)+it$ of $\xi^\w$ will be represented as a horizontal
line-segment between $x+it$ and $(x+1)+it$ and called a `bridge'.  
The interpretation of these objects will
differ slightly for the three different representations, as we now
describe.  See
Figures~\ref{fk-rep-fig} and~\ref{kw_fig} for examples.

\subsection{FK-representation}
\label{fk-sec}

For this representation we set $r^\b=h$ and $r^\w=2J$.  We interpret
the cuts $x+it\in\xi^\b$ as \emph{severing} a line-segment 
$x+i[0,\beta]$, and the
bridges $\xi^\w$ as \emph{connecting} neighbouring line segments.  Thus the
configuration $\xi$ is a  partly continuous
percolation-configuration.  The 
maximal connected subsets of $\Dom^\b$ are called \emph{components},
and their number is denoted $k^\b(\xi)$.  The components may be defined
with respect to various different boundary conditions, 
but for now we only consider the
`vertically periodic' boundary condition, meaning that the points at
the top  and bottom of $\Dom^\b$ are identified (i.e.\ we treat
$[0,\beta]$ as a circle).  
See Figure~\ref{fk-rep-fig}.

The \fk-representation expresses
\be\label{fk}
\cZ_{N,\beta}=e^{\beta J(N-1)} \EE_{h,2J}[2^{k^\b(\xi)}],
\quad
\el\s^{(3)}_x\s^{(3)}_y\er_{N,\beta}=
\frac{\EE_{h,2J}[\one\{x\lra y\}2^{k^\b(\xi)}]}
   {\EE_{h,2J}[2^{k^\b(\xi)}]},
\ee
where $\{x\lra y\}$ denotes the event that $x,y\in\{1,\dotsc,N\}$
belong to the same connected component.

\begin{figure}[hbt]
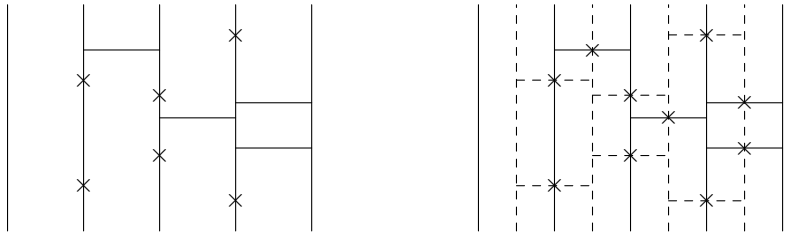

\centering
\includegraphics{fermionic.44}\hspace{2cm}
\includegraphics{fermionic.45}
\caption{
\emph{Left:}  Illustration of the \fk-representation.  Cuts ($\times$)
disconnect, bridges (horizontal line segments) connect, and top and
bottom of the intervals are identified.  The number $k(\xi)$ of
components is 5.
\emph{Right:}  The same \fk-sample $\xi$ (solid) with its dual $\xi'$ 
(dashed).  
}
\label{fk-rep-fig}
\end{figure}

With an \fk-configuration $\xi$ we can associate a \emph{dual}
configuration $\xi'$, whose connected components are subsets of
$\Dom^\w$ rather than $\Dom^\b$.  For simplicity we describe this in
the case when $\xi^\b$ has no cuts on the left- or rightmost intervals
$1+i[0,\beta]$ and $N+i[0,\beta]$.  We obtain $\xi'$ by drawing a
bridge from $(x-\sfrac12)+it$ to $(x+\sfrac12)+it$ for each cut
$x+it\in\xi^\b$, and placing a cut $\times$ at $(x+\sfrac12)+it$
whenever $\xi^\w$ has a bridge there.   See Figure~\ref{fk-rep-fig}. 
Objects, such as cuts, bridges and components, pertaining to $\xi'$
will be referred to as \emph{dual} and those of $\xi$ as \emph{primal}
when a distinction needs to be made.  The number of dual components
will be denoted $k^\w(\xi)$.
It turns out that $\xi'$ also has the law of a \fk-configuration, with
adjusted parameters.
We will return to this construction when we define the \fk-observable
in Section~\ref{fk-obs-sec}.

\subsection{Random-parity representation}

For this representation we set $r^\b=0$ and $r^\w=J$, thus there are 
only bridges.  We use auxiliary configurations $\psi\in\{0,1\}^N$
together with a fixed, finite subset $A\se\Dom^\b$ of 
\emph{sources}.  The
configuration $\psi$ is extended to a function
$\psi_A:\Dom^\b\to\{0,1\}$, 
in a way  which depends on $\xi^\w$ and $A$, using the following
rules. 
The function $\psi_A(x+it)$ is equal to $\psi(x)$ for $t$ from 0 to the
first time of either a bridge $(x\pm\sfrac12)+it\in\xi^\w$, 
\emph{or}  a source $x+it\in A$.  At such a point it switches
to $1-\psi(x)$.  Then it stays at that value until it encounters
another bridge-endpoint or source, where it switches back to
$\psi(x)$; and so 
on.  See Figure~\ref{kw_fig} for an example.  

The subset of $\Dom^\b$ where $\psi_A$ takes value 1 is denoted
$I(\psi_A)=\psi_A^{-1}(1)$, and will for definiteness be taken to be
closed.   We denote its total length  $|I(\psi_A)|$.  We will only be
considering the cases when either $A=\es$ or $A$ consists of two
points;  in the former case $I(\psi_A)$ consists of a collection of
loops, in the latter case loops plus a unique path connecting the two
points of $A$.

We impose  the periodicity constraint that
$\psi(x+i\beta)=\psi(x)$ for all $x\in\{1,\dotsc,N\}$;  
if $x\in A$ then the correct
interpretation is $\psi(x+i\beta)=1-\psi(x)$ due to the
switching-rule.  
Hence we discount some
configurations $\xi$, specifically those where some line
$x+i[0,\beta]$ meets an odd number of switching-points. 
As we will see presently, this discounting  can be done
formally by redefining
$|I(\psi_A)|=\oo$ when the constraint is violated.

The random-parity representation expresses
\be\begin{split}\label{rpr}
&\cZ_{N,\beta}=e^{\beta h N+\beta J(N-1)}
\EE_{0,J}\Big[\sum_{\psi\in\{0,1\}^N}\exp(-2h|I(\psi_\es)|)\Big],\\
&\el\s^{(3)}_x\s^{(3)}_y\er_{N,\beta}=\frac{
\EE_{0,J}\Big[\sum_{\psi\in\{0,1\}^N}\exp(-2h|I(\psi_{\{x,y\}})|)\Big]
}
{\EE_{0,J}\Big[\sum_{\psi\in\{0,1\}^N}\exp(-2h|I(\psi_\es)|)\Big]}.
\end{split}\ee
This representation is a quantum version of Aizenman's
random-current representation~\cite{aiz82}.
There is a notion of planar duality also for this representation, 
mapping onto the space--time spin representation, which we describe
now.  

\subsection{Space--time spin representation}

This representation plays a less prominent role in this note, 
and is mainly interesting since it is dual to the random-parity
representation.  
We now set $r^\b=h$ and $r^\w=0$, thus there
are only cuts.  We let $\S(\xi)$ denote the set of functions
$\s:\Dom^\b\to\{-1,+1\}$ which are constant between points of $\xi^\b$,
change value at the points of $\xi^\b$, and satisfy the periodicity
constraint $\s(x)=\s(x+i\beta)$ for all $x\in\{1,\dotsc,N\}$.  
See Figure~\ref{kw_fig}.
(For definiteness we may take $\s^{-1}(+1)$ to be closed;  also note
that for some $\xi$ we have $\S(\xi)=\es$.)

For readability we also write $\s_x(t)$ for $\s(x+it)$.
The space--time spin representation expresses
\be\label{stim}
\begin{split}
&\cZ_{N,\beta}=e^{\beta h N}
\EE_{h,0}\Big[
\sum_{\s\in\S(\xi)}\exp\Big( J\sum_{z=1}^{N-1}\int_0^\beta
\s_z(t)\s_{z+1}(t)\, dt\Big)\Big],\\
&\el\s^{(3)}_x\s^{(3)}_y\er_{N,\beta}=\frac{
\EE_{h,0}\Big[\sum_{\s\in\S(\xi)}
\s(x)\s(y)\exp\Big( J\sum_{z=1}^{N-1}\int_0^\beta
\s_z(t)\s_{z+1}(t)\, dt\Big)\Big]
}{
\EE_{h,0}\Big[
\sum_{\s\in\S(\xi)}\exp\Big( J\sum_{z=1}^{N-1}\int_0^\beta
\s_z(t)\s_{z+1}(t)\, dt\Big)\Big]}.
\end{split}
\ee

\subsection{Kramers--Wannier duality}

We now describe a duality between the random-parity and
spin-representations.  We will associate (in a reversible way) to a
spin-configuration $\s:\Dom^\b\to\{-1,+1\}$ a `dual'
random-parity-configuration 
$\psi=\psi_\es:\Dom^\w\to\{0,1\}$.  Note that the domain of  $\psi$ is
$\Dom^\w$ rather than $\Dom^\b$.  
We impose the `wired' boundary condition 
\be\begin{split}
\s(1+it)=\s(N+it)=\s(x)&=\s(x+i\beta)=+1,\\
&\mbox{for all } t\in[0,\beta], x\in\{1,\dotsc,N\}.
\end{split}\ee
As we will see, this will automatically lead to the boundary condition 
\be
\psi(x+\sfrac12)=\psi((x+\sfrac12)+i\beta)=0,
\mbox{ for all } x\in\{1,\dotsc,N\}.
\ee
Subject to the boundary conditions, 
the sums over $\s$ in~\eqref{stim}
and $\psi$ in~\eqref{rpr}
contribute with at most one nonzero term each, hence they will not be
written out. 

\begin{figure}[hbt]
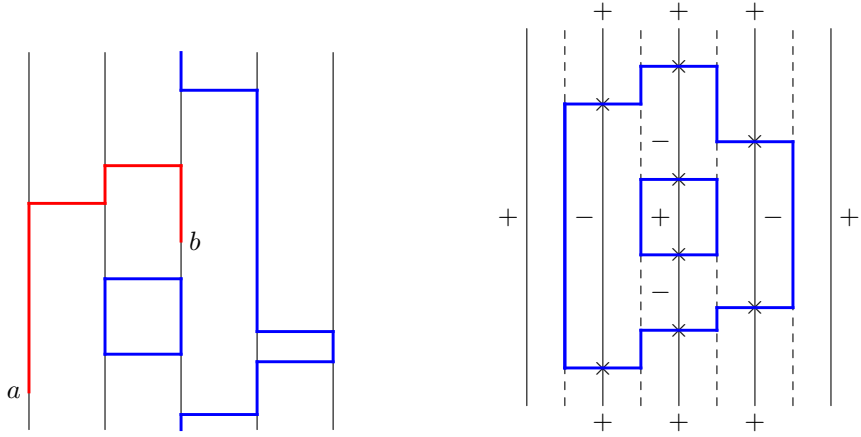

\centering
\includegraphics{fermionic.46}\hspace{2cm}
\includegraphics{fermionic.4}
\caption{
\emph{Left:}  Sample of the random-parity representation with source
set $A=\{a,b\}$.  
Intervals where $\psi=1$ are drawn bold, with red
for the unique path between $a$ and $b$ and blue for the loops.
\emph{Right:} Duality between the space--time spin and random-parity
  representations.  Values $+$ and $-$ indicate the value of $\s(z)$
on the corresponding interval in $\Dom^\b$, 
and these values flip at cuts $\times$.  
Blue vertical intervals mark where $\psi(z)=1$.}
\label{kw_fig}
\end{figure}

We construct $\psi$ from $\s$ as follows, see Figure~\ref{kw_fig}. 
If two neighbouring points $x+it$ and
$(x+1)+it$ have the same spin-value, 
$\s(x+it)=\s((x+1)+it)$, then we set
$\psi((x+\sfrac12)+it)=0$; 
otherwise if   $\s(x+it)\neq\s((x+1)+it)$,
then we set $\psi((x+\sfrac12)+it)=1$.
If $x+it\in \xi^\b$ is a
point of spin-flip for $\s$, we draw a bridge between
$(x-\sfrac12)+it$ 
and $(x+\sfrac12)+it$.  Thus the bridges form a Poisson
process of  rate $h$.  

Writing $\cZ^+_{N,\beta}(h,J)$ for the partition function~\eqref{stim}
associated with the spin-configurations, 
we have that
\be\begin{split}
\cZ^+_{N,\beta}(h,J)
&=e^{\beta h N}\EE_{h,0}\Big[
\exp\Big(J\sum_{x=1}^{N-1}\int_0^\beta
\s_x(t)\s_{x+1}(t)\, dt\Big)\Big]\\
&=e^{\beta h N}\EE_{h,0}\Big[
\exp\Big(J\sum_{x=1}^{N-1}\int_0^\beta
[1-2\psi((x+\sfrac12)+it)]\, dt\Big)\Big]\\
&= e^{\beta h N+\beta J (N-1)}
\EE_{h,0}\big[\exp(-2J |I(\psi)|)\big].
\end{split}\ee
Comparing with~\eqref{rpr}, we see that the last factor
\be
\EE_{h,0}\big[\exp(-2J |I(\psi)|)\big]=
e^{-\beta J(N-1)-\beta h (N-2)} 
\cZ^{0}_{N-1,\beta}(J,h),
\ee
where $\cZ^{0}_{N-1,\beta}(J,h)$ is the partition function associated with
the $\psi$:s with the prescribed boundary condition.  Note that the
order of the parameters $h,J$ is swapped.

We conclude that
\be
\cZ^+_{N,\beta}(h,J)=e^{2\beta h} \cZ^0_{N-1,\beta}(J,h).
\ee
Assuming (as can be justified) the existence of the limit as well as its
independence of the boundary condition, we deduce that the free energy 
$f(h,J)=\lim_{N,\beta\to\oo} \tfrac{1}{\beta N} \log \cZ_{N,\beta}(h,J)$
satisfies $f(h,J)=f(J,h)$.
This symmetry is consistent with a phase-transition at
$h=J$.  
In the rest of this note we consider only the critical case, 
$h=J$.

\section{S-holomorphic functions}
\label{s-hol-sec}

\subsection{Discrete domains}
\label{dom-sec}

As indicated above, we will be considering functions on
(bounded subsets of) $\d\ZZ+i\RR\se\CC$.  
We use the notation
\[
\CC_\d^\b=\d \ZZ+ i\RR,\quad 
\CC_\d^\w=\CC_\d^\b+\sfrac\d2, \quad \mbox{and}\quad 
\CC^\diamondsuit_\d=(\CC^\b_\d\cup\CC^\w_\d)+\sfrac\d4.
\]
We will sometimes refer to points of $\CC^\b_\d$ as
\emph{primal} or \emph{black}, points of $\CC^\w_\d$ as
\emph{dual} or \emph{white}, and points of 
$\CC^\diamondsuit_\d$ as \emph{medial}.
See Figure~\ref{dom-fig}
for illustrations of the definitions that follow.

Let $\partial_\d:[0,1]\to\CC$ be
a simple closed rectangular path, consisting of
vertical and horizontal line segments,  whose
vertical segments are restricted to $\CC_\d^\b$.
Let $\Dom_\d$ denote the bounded component of
$\CC\sm\partial_\d[0,1]$.  
Such a domain $\Dom_\d$ will be referred to as
a \emph{primal (discrete) domain}.  We also write,
for $\ast\in\{\b,\w\}$,
\be\begin{split}\label{doms-not}
\Dom_\d^\ast=\ol{\Dom_\d}\cap\CC^\ast_\d,\quad
\partial\Dom_\d^\ast=\Dom_\d^\ast\cap\partial\Dom_\d,\quad
\Dom_\d^{\ast,\itr}=\Dom_\d^\ast\sm\partial\Dom_\d^\ast.
\end{split}\ee
Note that $\Dom_\d^\ast$ consists of a collection of
vertical line segments, and $\partial\Dom_\d^\ast$ of
vertical line segments together with a finite number
of points (forming the hortizontal part of the boundary).
We similarly define a \emph{dual (discrete) domain} $\Dom_\d$
by shifting the above definition by
$\sfrac\d2$ (thus swapping $\CC^\b_\d$ and $\CC^\w_\d$).

We will also consider Dobrushin domains.  For this we let 
$a_\d,b_\d\in\CC_\d^\diamondsuit$ be two distinct medial points, and
let $\partial_\d:[0,1]\to\CC$ be a simple closed 
\emph{positively  oriented} rectangular path, satisfying
\[
\partial_\d(0)=\partial_\d(1)=a_\d,\quad
\partial_\d(\sfrac12)=b_\d.
\]
We define $\partial_\d^\b,\partial_\d^\w:[0,1]\to\CC$ by
\[
\partial_\d^\w(t)=\partial_\d(t/2),\quad 
\partial_\d^\b(t)=\partial_\d(1-t/2),\quad t\in[0,1].
\]
Thus $\partial_\d^\w$ goes from $a_\d$ to $b_\d$ in the
counter-clockwise direction, and 
$\partial_\d^\b$ goes from $a_\d$ to $b_\d$ in the
clockwise direction.
Finally we assume that the vertical segments of 
$\partial_\d^\b$ and $\partial_\d^\w$ belong to 
$\CC^\b_\d$ and $\CC^\w_\d$, respectively.
Again we write $\Dom_\d$ for the bounded component of 
$\CC\sm\partial_\d[0,1]$,  and we refer to the triple 
$(\Dom_\d,a_\d,b_\d)$ as a \emph{discrete Dobrushin domain}.
We define $\Dom_\d^\b$, $\partial\Dom_\d^\b$, 
$\Dom_\d^{\b,\itr}$, as well as
$\Dom_\d^\w$, $\partial\Dom_\d^\w$, 
$\Dom_\d^{\w,\itr}$, as in~\eqref{doms-not}.

\begin{figure}[hbt]
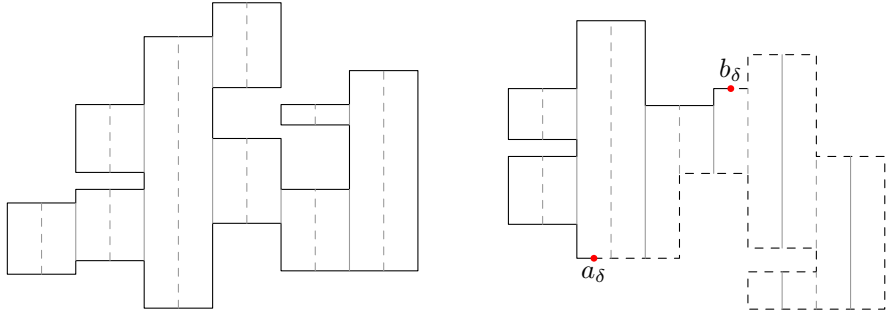

\centering
\includegraphics{fermionic.1}\hspace{1cm}
\includegraphics{fermionic.2}

\caption{
\emph{Left:} 
A primal domain $\Dom_\d$.  The boundary is drawn with solid
  black lines, while $\Dom_\d^\b$ consists of the solid black and gray
  vertical lines and $\Dom_\d^\w$ of the dashed gray vertical lines.
\emph{Right:} A Dobrushin domain $\Dom_\d$ with $\partial^\b_\d$ drawn
  solid and $\partial^\w_\d$ dashed.
}
\label{dom-fig}
\end{figure}

For a primal, dual or Dobrushin domain
$\Dom_\d$, and $\ast\in\{\b,\w\}$, 
we define the \emph{vertical} 
and \emph{horizontal} parts of the boundary $\partial\Dom_\d^\ast$ by 
\be\begin{split}
\partial^\rv\Dom^\ast_\d&=\{z\in\partial\Dom^\ast_\d:
z+\eps\not\in\Dom_\d\mbox{ or } z-\eps\not\in\Dom_\d
\mbox{ for small enough }\eps>0\},\\
\partial^\rh\Dom^\ast_\d&=\{z\in\partial\Dom^\ast_\d:
z+ i\eps\not\in\Dom_\d\mbox{ or }
z- i\eps\not\in\Dom_\d\mbox{ for small enough }\eps>0\}.
\end{split}\ee
We also let 
$\partial^\rv\Dom_\d=\partial^\rv\Dom^\b_\d\cup \partial^\rv\Dom^\w_\d$
and
$\partial^\rh\Dom_\d=\partial^\rh\Dom^\b_\d\cup \partial^\rh\Dom^\w_\d$.
In words, $\partial^\rv\Dom_\d$ consists of the vertical 
segments of $\partial\Dom_\d$, and $\partial^\rh\Dom_\d$ of the
endpoints of segments in $\ol{\Dom_\d}$.
We finally make the assumption on $\Dom_\d$ that if
$z\in\partial^\rv\Dom_\d$ then at least one of $z\pm\sfrac\d2$ belongs
to the interior $\Dom_\d^{\b,\itr}\cup\Dom_\d^{\w,\itr}$.

In what follows we will consider triples $(\Dom_\d,a_\d,b_\d)$
which are either discrete Dobrushin domains,
alternatively discrete 
primal or dual domains with two marked points 
$a_\d,b_\d\in\partial\Dom_\d$.
One may think of these as approximating 
a simply connected domain $\Dom\se\CC$ with 
two marked points $a,b$ on its boundary.

\subsection{S-holomorphic functions}

Let $\Dom_\d$ be a discrete domain, as above, and
$F:\Dom_\d\to\CC$ a function.  We will be using the notation 
\be
\dot F(z):=\lim_{\eps\to 0}\frac{F(z+i\eps)-F(z)}{\eps},
\quad \mbox{with }\eps\in\RR,
\ee
for the derivative of $F$ in the `vertical' direction, when it exists.
We similarly write $\ddot F(z)$ for the second derivative.

For a complex number $\zeta$, with $|\zeta|=1$, 
and $z\in\CC$, we write 
\be\label{proj-def}
\mathrm{Proj}[z;\zeta]=
\mathrm{Proj}[z;\zeta\RR]=\tfrac12(z+\ol z \zeta^2)
\ee
for the projection of $z$ onto (the straight line through 0 and)
$\zeta$.  The cases when $\zeta=e^{\pm i\pi/4}$ will be particularly
important in what follows, and 
we will write $\ell(\up)=e^{-i\pi/4}\RR$ and
$\ell(\dn)=e^{i\pi/4}\RR$.  (This choice of notation will be motivated
below, in the context of the \fk-observable).
We define 
\be\label{up-dn-notation}
F^\up(z)=\mathrm{Proj}[F(z);\ell(\up)],\quad
F^\dn(z)=\mathrm{Proj}[F(z);\ell(\dn)].
\ee
Note that $F(z)=F^\up(z)+F^\dn(z)$ since 
$\ell(\up)\perp\ell(\dn)$.

\begin{figure}[hbt]
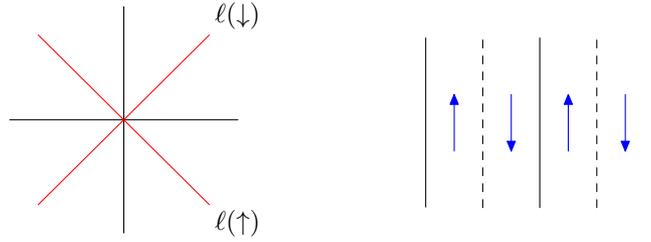

\centering
\includegraphics{fermionic.42}\hspace{2cm}
\includegraphics{fermionic.43}
\caption{
\emph{Left:}  The lines $\ell(\up)$ and $\ell(\dn)$.
\emph{Right:}  Illustration of conditions~\eqref{s-hol-def-w-1}
and~\eqref{s-hol-def-u-1} in Definition~\ref{s-hol-def}.  For a pair
of adjacent black and white points, 
separated by an arrow in direction $\a\in\{\up,\dn\}$,
the projections of $F$ onto $\ell(\a)$ are the same. 
}
\label{s-hol-fig}
\end{figure}

\begin{definition}[s-holomorphic]\label{s-hol-def}
A function $F:\Dom^\b_\d\cup\Dom^\w_\d\to\CC$ is s-holomorphic at a
point  $w\in\Dom_\d^{\w,\itr}$ if the following hold:
\be\label{s-hol-def-w-1}
F^\up(w)=F^\up(w-\sfrac\d2),\quad
F^\dn(w)=F^\dn(w+\sfrac\d2),\quad\mbox{and}
\ee
\be\label{s-hol-def-w-2}
\begin{split}
\dot F^\up(w)&=\tfrac i\d 
\big(F^\dn(w+\sfrac\d2)-F^\dn(w-\sfrac\d2)\big),\\
\dot F^\dn(w)&=\tfrac i\d 
\big(F^\up(w+\sfrac\d2)-F^\up(w-\sfrac\d2)\big).
\end{split}
\ee
It is s-holomorphic at a point
$u\in\Dom_\d^{\b,\itr}$ if the following hold:
\be\label{s-hol-def-u-1}
F^\up(u)=F^\up(u+\sfrac\d2),\quad
F^\dn(u)=F^\dn(u-\sfrac\d2),\quad\mbox{and}
\ee
\be\label{s-hol-def-u-2}
\begin{split}
\dot F^\up(u)&=\tfrac i\d 
\big(F^\dn(u+\sfrac\d2)-F^\dn(u-\sfrac\d2)\big),\\
\dot F^\dn(u)&=\tfrac i\d 
\big(F^\up(u+\sfrac\d2)-F^\up(u-\sfrac\d2)\big).
\end{split}
\ee
If $F$ is s-holomorphic at every point 
$z\in\Dom_\d^{\b,\itr}\cup\Dom_\d^{\w,\itr}$ then we simply say that
$F$ is s-holomorphic in $\Dom_\d$.  
\end{definition}

The choice of the term s-holomorphic is mainly motivated by
Proposition~\ref{H-prop} below, which is completely analogous to the
classical case (e.g.\ Proposition~3.6 of~\cite{ch-sm}).

It is easy to see that a function $F$ which is s-holomorphic at a point
$z\in\Dom_\d^{\b,\itr}\cup\Dom_\d^{\w,\itr}$
satisfies the following natural preholomorphicity condition:
\be\label{prehol}
\tfrac1\d\big(F(z+\sfrac\d2)-F(z-\sfrac\d2)\big)
+i\dot F(z)=0.
\ee
However, as for the classical case, the main benefit of s-holomorphic
functions $F$ is that they have well-behaved discrete analogs of 
$\Im\big(\int F^2\big)$.  
In the next result we write $\D_\d$ for the appropriate Laplacian
operator given by
\be\label{lap-def}
[\D_\d f](z)=\ddot f(z)
+\tfrac1{\d^2}\big(f(z+\d)+f(z-\d)-2f(z)\big).
\ee
We say that a function $h$ is $\D_\d$-harmonic (respectively,
$\D_\d$-sub- or $\D_\d$-super-harmonic) at a point
$z\in\CC_\d^\b\cup\CC_\d^\w$ if  
$[\D_\d h](z)=0$ (respectively, $[\D_\d h](z)\geq0$ or $[\D_\d h](z)\leq0$).

\begin{proposition}\label{H-prop}
Let $F$ be s-holomorphic in $\Dom_\d$.  Then there
is a function 
$H:\Dom_\d^\b\cup\Dom_\d^\w\to\RR$, unique
up to an additive constant, satisfying the following.
Firstly, for $z$ s.t.\ $[z,z+\sfrac\d2]\se\Dom_\d$,
\be\label{discr-diff-eq}
H(z+\sfrac\d2)-H(z)=
\left\{
\begin{array}{ll}
+|F^\dn(z)|^2, & \mbox{if } z\in\Dom_\d^\w,\\
-|F^\up(z)|^2, & \mbox{if } z\in\Dom_\d^\b,
\end{array}
\right.
\ee
and secondly, for any $z\in\Dom_\d^\b\cup\Dom_\d^\w$,
\be\label{der-eq}
\dot H(z)=\tfrac2\d F^\up(z) F^\dn(z).
\ee
Moreover, we have for all $u\in\Dom_\d^{\b,\itr}$ and 
$w\in\Dom_\d^{\w,\itr}$ that
\be\label{H-sub-sup}
[\D_\d H](u)=|\dot F(u)|^2
\quad\mbox{and}\quad
[\D_\d H](w)=-|\dot F(w)|^2.
\ee
Hence 
$H$ is $\D_\d$-sub-harmonic in $\Dom_\d^{\b,\itr}$ and 
$\D_\d$-super-harmonic in $\Dom_\d^{\w,\itr}$.
\end{proposition}

\begin{remark}
The function $\d H(z)$ is a discrete analog 
of $\Im\big(\int^z F^2\big)$.
Indeed, since
$2F^\up(z)F^\dn(z)=\Re[F(z)^2],$ we see that
if $u,u'\in\Dom_\d^\b$ with 
$u=x+iy$ and $u'=x+iy'$ for some $x\in\d\ZZ$ and $y,y'\in\RR$ such
that $[u,u']\se\Dom_\d^\b$, then
\be
\begin{split}
H(u')-H(u)&=\int_{y}^{y'} \dot H(x+it) \,dt
=\frac1\d\int_y^{y'} \Re[F(x+it)^2] \,dt\\
&=\frac1\d\int_y^{y'} \Im[iF(x+it)^2] \,dt
=\tfrac1\d\Im\Big[\int_u^{u'} F(z)^2 \,dz\Big].
\end{split}
\ee
Similarly, if $v=u+\d\in\Dom_\d^\b$ 
and $w=u+\sfrac\d2=v-\sfrac\d2$ is midway between $u$ and $v$
then, using
$F^\up(z)^2+F^\dn(z)^2=i\,\Im[F(z)^2]$, we have
\be
\begin{split}
H(v)-H(u)&=H(w+\sfrac\d2)-H(w)+H(w)-H(w-\sfrac\d2)\\
&=|F^\dn(w)|^2-|F^\up(w)|^2
=\tfrac1i(F^\up(w)^2+F^\dn(w)^2)\\
&=\Im[F(w)^2].
\end{split}
\ee
\end{remark}

\begin{proof}[Proof of Proposition~\ref{H-prop}]
Uniqueness up to an additive constant follows since if
we fix $H(u)$ for some point $u$, then for $v\neq u$
we may obtain the value $H(v)$ by integrating
using~\eqref{discr-diff-eq} and~\eqref{der-eq}.
To see that $H$ is well-defined,
consider a situation such as in Figure~\ref{H-def-fig}.
It suffices to show that the total increment of
$H$ around the blue (left) contour
and around the green (right) contour are both equal to 0.
We prove this for the green (right) contour, the other one being
similar.  
\begin{figure}[hbt]
\centering
\includegraphics{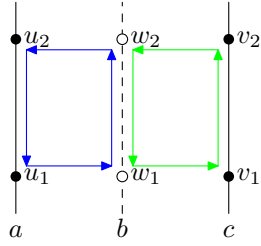}
\caption{Contours in the proof of Proposition~\ref{H-prop}.}
\label{H-def-fig}
\end{figure}

Let us write, for $a,b,c,t_1,t_2\in\RR$ and $j=1,2$,
 $u_j=a+i t_j$,  $w_j=b+i t_j$ and  $v_j=c+i t_j$.  
 We have
\be\begin{split}
[H(v_2)-&H(v_1)]+[H(w_1)-H(w_2)]=
\int_{t_1}^{t_2}\big(\dot H(c+it)-\dot H(b+it)\big)dt\\
&=\tfrac2\d\int_{t_1}^{t_2}F^\dn(b+it)\big(
F^\up(c+it)-F^\up(b+it)\big)dt\\
&=\tfrac1i\int_{t_1}^{t_2} 2F^\dn(b+it)\dot F^\dn(b+it) dt\\
&=\tfrac1i\big(F^\dn(w_2)^2-F^\dn(w_1)^2\big)
=|F^\dn(w_2)|^2-|F^\dn(w_1)|^2\\
&=[H(v_2)-H(w_2)]-[H(v_1)-H(w_1)].
\end{split}\ee
That is, the increments around the green contour satisfy
\[
[H(v_2)-H(v_1)]+[H(w_2)-H(v_2)]+
[H(w_1)-H(w_2)]+[H(v_1)-H(w_1)]=0,
\]
as required.

We turn now to the statement~\eqref{H-sub-sup}.
We give the details for $u\in\Dom_\d^\b$, the case $w\in\Dom_\d^\w$
being similar.
Since $\dot H(u)=\tfrac2\d F^\up(u)F^\dn(u)$
we have 
\be
\ddot H(u)=\tfrac2\d\big(
\dot F^\up(u)F^\dn(u)+F^\up(u)\dot F^\dn(u)
\big).
\ee
Using s-holomorphicity we deduce that
\be
\begin{split}
\ddot H(u)&=2\tfrac{i}{\d^2}\big(
[F^\dn(u+\sfrac\d2)-F^\dn(u-\sfrac\d2)]F^\dn(u)\\
&\qquad\qquad
+[F^\up(u+\sfrac\d2)-F^\up(u-\sfrac\d2)]F^\up(u)
\big)\\
&=\tfrac{i}{\d^2} \big(
2F^\dn(u-\sfrac\d2)F^\dn(u+\sfrac\d2)-2F^\dn(u-\sfrac\d2)^2\\
&\qquad\qquad+2F^\up(u+\sfrac\d2)^2-
2F^\up(u-\sfrac\d2)F^\up(u+\sfrac\d2)
\big).
\end{split}
\ee
Next, 
\be
\begin{split}
H(u-\d)-H(u)&=
|F^\up(u-\sfrac\d2)|^2-|F^\dn(u-\sfrac\d2)|^2\\
&=i(F^\up(u-\sfrac\d2)^2+F^\dn(u-\sfrac\d2)^2),\mbox{ and }\\
H(u+\d)-H(u)&=
|F^\dn(u+\sfrac\d2)|^2-|F^\up(u+\sfrac\d2)|^2\\
&=-i(F^\dn(u+\sfrac\d2)^2+F^\up(u+\sfrac\d2)^2).
\end{split}
\ee
It follows that
\be
\begin{split}
\d^2&\ddot H(u)+
[H(u-\d)-H(u)]+[H(u+\d)-H(u)]\\
&=i\big(
[F^\up(u-\sfrac\d2)-F^\up(u+\sfrac\d2)]^2-
[F^\dn(u-\sfrac\d2)-F^\dn(u+\sfrac\d2)]^2
\big).
\end{split}
\ee
Writing $F^\up(u-\sfrac\d2)=ae^{-i\pi/4}$,
$F^\up(u+\sfrac\d2)=be^{-i\pi/4}$, 
$F^\dn(u-\sfrac\d2)=ce^{i\pi/4}$
and $F^\dn(u+\sfrac\d2)=de^{i\pi/4}$,
for $a,b,c,d\in\RR$, the right-hand-side
equals
\be
i\big(\tfrac1i[a-b]^2-i[c-d]^2
\big)=(a-b)^2+(c-d)^2.
\ee
But we also have that
\be\begin{split}
&|F(u-\sfrac\d2)-F(u+\sfrac\d2)|^2\\
&\quad=
|(F^\up(u-\sfrac\d2)-F^\up(u+\sfrac\d2))+
 (F^\up(u-\sfrac\d2)-F^\up(u+\sfrac\d2))|^2\\
&\quad=(a-b)^2+(c-d)^2, \mbox{ since } \ell(\up)\perp\ell(\dn).
\end{split}\ee
Thus, using also~\eqref{prehol},
\be
\d^2 [\D_\d H](u)=
|F(u-\sfrac\d2)-F(u+\sfrac\d2)|^2=\d^2|\dot F(u)|^2,
\ee
as claimed.
\end{proof}

\section{The FK-observable}\label{fk-obs-sec}

\subsection{Definition}

Let $(\Dom_\d,a_\d,b_\d)$ be a Dobrushin domain
(see Section~\ref{dom-sec} for notation).
We will consider \fk-configurations $\xi$ in $\Dom_\d$ and 
their duals $\xi'$.  These are defined as in
Section~\ref{fk-sec}
with some adaptations of the boundary condition.
We take $\xi=\xi^\b\cup\xi^\w$ with 
$\xi^\b\se\Dom^{\b,\itr}_\d$ and $\xi^\w\se\Dom^{\w,\itr}_\d$ finite
subsets.   Note that we do not allow $\xi^\b$ to have any points on
the black part $\partial_\d^\b$ of the boundary, nor do we allow
$\xi^\w$ to have any points on the white part $\partial^\w_\d$. 
Instead of applying periodic boundary conditions, we let
horizontal 
segments in $\partial^\b_\d$ and $\partial^\w_\d$ count as primal and
dual bridges, respectively.  Thus, in essence, we have separately
wired together the black and white parts $\partial^\b_\d$ and
$\partial^\w_\d$ of the boundary.
See Figure~\ref{domain-interface-fig}.

\begin{figure}[hbt]
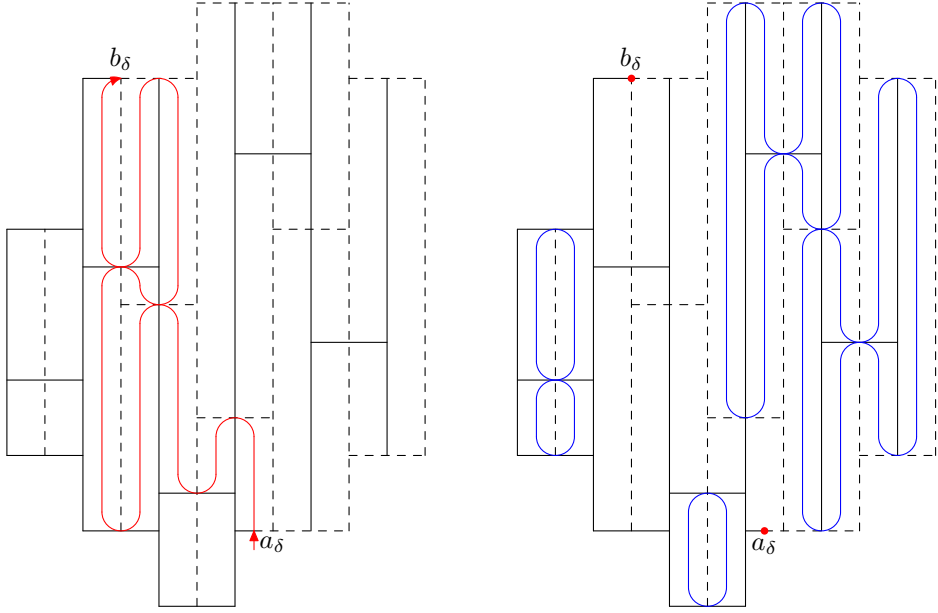

\centering
\includegraphics{fermionic.5}\hspace{1cm}
\includegraphics{fermionic.6}
\caption{Dobrushin domain $(\Dom_\d,a_\d,b_\d)$ 
with an \textsc{fk} configuration $\xi$ and its dual $\xi'$, 
as well as the interface $\g$
(left) and the $L(\xi)=5$  loops (right).  
We have omitted the $\times$-marks for
cuts.}
\label{domain-interface-fig}
\end{figure}

We adjust the
locations of the points $a_\d$ and $b_\d$ slightly compared to
Section~\ref{dom-sec}, as follows.
Firstly, we assume that $a_\d$ is placed 
so that the first point of $\CC_\d^\b\cup\CC_\d^\w$ visited by
$\partial_\d^\b$ (as it travels clockwise from $a_\d$ to $b_\d$)
belongs to $\CC_\d^\b$.  Thus $a_\d$ is of the form $u+\sfrac\d4$
for some $u\in\Dom_\d^\b$ if $\Dom_\d$ is `above' $a_\d$, or of the
form $u-\sfrac\d4$ if $\Dom_\d$ is `below' $a_\d$.
With this assumption,  an \fk-configuration $\xi$ together
with its dual $\xi'$ define an interface $\g$ from $a_\d$ to $b_\d$,
separating the (primal) component of $\partial_\d^\b$ from the (dual)
component of $\partial_\d^\w$, and $\g$ always has black on the left and
white on the right as it travels from $a_\d$ to $b_\d$.
We take $\g$ to travel in the directions $\up,\dn$
on the medial lattice $\CC_\d^\diamondsuit$ between bridges, and
in the directions $\lt,\rt$ at bridges (if $\g$ passes the same bridge
twice we slightly separate the points where it passes).
We also shift $b_\d$ left or right by $\sfrac\d4$ so
that the interface $\g$ ends pointing in the direction $\rt$ into
$b_\d$.   See Figure~\ref{domain-interface-fig} again.

Apart from the interface $\g$, we also draw a loop around each 
(primal and dual) component which is disjoint from the boundary.
We let $L(\xi)$ denote the number of such loops.

Let $\EE_\d(\cdot)$ denote the probability measure under which 
$\xi^\b$ and $\xi^\w$ are independent Poisson processes on 
$\Dom_\d^{\b,\itr}$ and $\Dom_\d^{\w,\itr}$, respectively, 
\emph{both with the same rate} $\tfrac1{\d\sqrt2}$.  By~\eqref{fk}, the
appropriate density of a random \fk-configuration $\xi$ with respect
to $\EE_\d(\cdot)$ is proportional to
\be\label{fk-dens-1}
2^{k^\b(\xi)} h^{|\xi^\b|} (2J)^{|\xi^\w|} (\d\sqrt2)^{|\xi^\b|+|\xi^\w|}.
\ee
Using the Euler-relation one may see that
$k^\b(\xi)-|\xi^\b|=k^\w(\xi)-|\xi^\w|+\mathrm{cst}$ for some constant
not depending on $\xi$.  Also, $L(\xi)=k^\b(\xi)+k^\w(\xi)-2$.
We choose the parameters
\be
h=J=\tfrac1{2\d}.
\ee
 It then follows that the density~\eqref{fk-dens-1} is 
proportional to simply $(\sqrt2)^{L(\xi)}$.  
We write $\hat\EE_\d=\hat\EE_{(\Dom_\d,a_\d,b_\d)}$  
for the critical \fk-law in $(\Dom_\d,a_\d,b_\d)$ given by
\be\label{fk-dens}
\frac{d\hat\EE_\d}{d\EE_\d}(\xi)=
\frac{(\sqrt{2})^{L(\xi)}}{Z_\d},
\mbox{ where } Z_\d=\EE_\d[(\sqrt2)^{L(\xi)}].
\ee

Now let $z\in\Dom_\d^\b\cup \Dom_\d^\w$ be arbitrary.
For $\a\in\{\up,\dn,\lt,\rt\}$, define the event
\[
\G^\a_z=\{\xi: \g(\xi)\mbox{ passes by $z$ in direction }\a\}.
\]
For $\a\in\{\up,\dn\}$ we count both the case when $\g$ passes
on the left side of $z$ (i.e.\ goes through $z-\sfrac\d4$) and 
when it passes on the right side (i.e.\ goes through $z+\sfrac\d4$).
Similarly, for $\a\in\{\lt,\rt\}$ we count both the cases when
$\g$ passes `just below' $z$ and `just above' $z$.

Assuming that $\G^\a_z$ happens, let 
$W^\a_\g(z)$ denote the winding-angle (in radians) of $\g$
from $z$ to the exit $b_\d$;
if $\g$ passes $z$ twice, in opposite directions, we count here the
winding angle from when it passes in direction $\a$.
Note that $W^\a_\g(z)$ is deterministic up to a multiple
of $2\pi$.

We define the four
(random) functions $\varphi^\up(\xi; z)$,
$\varphi^\dn(\xi; z)$,
$\varphi^\lt(\xi; z)$ and
$\varphi^\rt(\xi; z)$ by
\be\label{phi-om-def}
\varphi^\a(\xi; z)=\one_{\G^\a_z}(\xi)
\exp(\tfrac{i}{2}W^\a_{\g(\xi)}(z)).
\ee
Note that the supports of $\varphi^\lt(\xi; z)$ and of
$\varphi^\rt(\xi; z)$ are discrete sets contained in
$\xi\cup\partial^\rh\Dom_\d$, 
whereas the supports of $\varphi^\up(\xi; z)$ and
$\varphi^\dn(\xi; z)$ are disjoint from $\xi$.
Also note 
that if $u\in\Dom_\d^\b$ is black 
and $w=u+\sfrac\d2$ is the white neighbour
of $u$ on the right, then 
$\varphi^\up(\xi;u)=\varphi^\up(\xi;w)$,
whereas if $w'=u-\sfrac\d2$ is the white neighbour of $u$
on the left then 
$\varphi^\dn(\xi;u)=\varphi^\dn(\xi;w')$.
(Here we assume that $u\pm\sfrac\d2\in\Dom_\d^\w$ in the appropriate
cases.)

\begin{definition}\label{fk-obs-def}
Write 
\be\label{phi-updn}
\Phi_\d^\up(z)=\hat\EE_\d[\varphi^\up(\xi;z)],\quad
\Phi_\d^\dn(z)=\hat\EE_\d[\varphi^\dn(\xi;z)].
\ee
 We define the \textsc{fk}--Ising observable 
$F_\d(z)=F^{\mathrm{FK}}_\d(z)$ by
\be\label{fk-obs-eq}
F_\d(z)=\Phi_\d^\up(z)+\Phi_\d^\dn(z),\quad
z\in\Dom_\d^\b\cup\Dom_\d^\w.
\ee
\end{definition}

We remark that the notation used here is consistent with our previous 
notation~\eqref{up-dn-notation} for the projections 
$F^\up$, $F^\dn$ of a function
$F$ onto $\ell(\up)=e^{-i\pi/4}\RR$ and
$\ell(\dn)=e^{i\pi/4}\RR$, in the sense that
\[
F_\d^\up(z)=\Phi_\d^\up(z)\mbox{ and } F_\d^\dn(z)=\Phi_\d^\dn(z).
\]
Indeed, if we identify arrows 
$\a\in\{\up,\dn,\lt,\rt\}$ with complex numbers by
the rules
\be
\rt\,=1=i^0,\quad
\up\,=i=i^1,\quad
\lt\,=-1=i^2,\quad
\dn\,=-i=i^3,
\ee
then we have that 
\be
W_{\g(\xi)}^\a(z)=-\arg(\a)+2\pi n(\xi)
\ee
for some random $n(\xi)\in\ZZ$.  Thus
$\varphi^\a(\xi;z)$ is a real multiple of $\sqrt{\ol{\a}}$
i.e.\ belongs to $\ell(\a)$.
Note that the line  $\sqrt{\ol{\a}}\RR$ does not
depend on the choice of square-root.

\subsection{Comparison with isoradial graphs}
\label{isorad-sec}

For readers familiar with the work of Chelkak and 
Smirnov~\cite{ch-sm} on the
classical Ising model on isoradial graphs, the following brief
discussion may be useful.  Let $0<\eps\ll\d$ and consider a rhombic
tiling of $\CC$ where all the rhombi have two vertices in each of
$\CC^\b_\d$ and $\CC^\w_\d$, and acute angle $2\eps$,
as in Figure~\ref{squeeze-2_fig}.
This corresponds to an isoradial embedding of $\ZZ^2$ with common
radius $\tfrac\d{2\cos(\eps)}$ and vertices restricted to
$\CC^\b_\d$.   

\begin{figure}[hbt]
\centering
\includegraphics{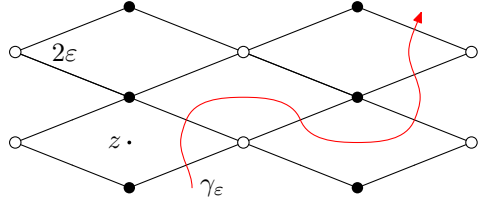}
\caption{Isoradial approximation of the \fk-representation of the
  \tfim.} 
\label{squeeze-2_fig}
\end{figure}

Let $\hat\EE_{\d,\eps}(\cdot)$ denote the law of the critical
(classical) \fk--Ising model in some Dobrushin-domain in this graph,
as given in eq.~(2.1) 
of~\cite{ch-sm}, and let $\g_\eps$ denote the interface.  
It is well-known that the laws $\hat\EE_{\d,\eps}$
converge weakly to $\hat\EE_\d$ as $\eps\to0$.

In this setting, the interface $\g_\eps$ is taken 
to cross the rhombus-sides perpendicularly,
i.e.\ roughly speaking in the directions
$\nearrow$, $\searrow$, $\nwarrow$
and $\swarrow$.  If we specify a rhombus as well as one of these four
directions of travel, this corresponds to a unique edge of the
rhombus, hence 
the edge-observables~\cite[eq.~(2.2)]{ch-sm} 
of Chelkak and Smirnov can be indexed
as $F_{\d,\eps}^\nearrow(z), F_{\d,\eps}^\searrow(z), \dotsc$
for  rhombus centres  $z$.  Using notation similar
to~\eqref{phi-om-def}, we have (up to a real factor)
\be
F_{\d,\eps}^\a(z)=\hat\EE_{\d,\eps}[\one_{\G_z^\a}
\exp(\tfrac{i}{2}W^\a_{\g_\eps}(z))],\quad
\a\in\{\nearrow,\searrow,\nwarrow,\swarrow\}.
\ee

We may further take $\g_\eps$
to pass `closest' to rhombus centres $z$
in the directions $\up$, $\dn$, $\lt$ or $\rt$.
This allows us to define more observables:
\be
\Phi_{\d,\eps}^\a(z)=\hat\EE_{\d,\eps}[\one_{\G_z^\a}
\exp(\tfrac{i}{2}W^\a_{\g_\eps}(z))],\quad
\a\in\{\up,\dn,\lt,\rt\}.
\ee
Clearly each $\Phi_{\d,\eps}^\a(z)\in\ell(\a)$
as is the case for the $\Phi_\d^\a(z)$ (provided we assume that
$\g_\eps$ exits the domain in the direction $\rt$).

Referring to Figure~\ref{squeeze-2_fig}, we see for example that 
if $\g_\eps$ enters the
rhombus of $z$ in direction $\nwarrow$
(edge on the lower right of $z$), then
it passes closest to $z$ in \emph{either}
direction $\up$ as depicted, or directon $\lt$,
\emph{but not both} ($\up$ if there is a black vertical 
edge at $z$, and $\lt$ if there is a white horizontal
edge).  
Similar considerations apply at all rhombus centres, and this allows
us to derive linear relations for the  
$F_{\d,\eps}^\a(z)$ in terms of the $\Phi_{\d,\eps}^\a(z)$.
Writing $\eps^\star=\tfrac\pi2-\eps$ we have:
\be\label{mtrx}
\begin{pmatrix}
F_{\d,\eps}^\nwarrow(z) \\
F_{\d,\eps}^\nearrow(z) \\
F_{\d,\eps}^\searrow(z) \\
F_{\d,\eps}^\swarrow(z) 
\end{pmatrix}=
\begin{pmatrix}
e^{-\tfrac i2 \eps} & 0 & 0 & e^{\tfrac i2 \eps^\star} \\
e^{\tfrac i2 \eps} & e^{-\tfrac i2 \eps^\star} & 0 & 0 \\
0 & e^{\tfrac i2 \eps^\star} & e^{-\tfrac i2 \eps} & 0 \\
0 & 0 & e^{\tfrac i2 \eps} & e^{-\tfrac i2 \eps^\star}
\end{pmatrix}
\begin{pmatrix}
\Phi_{\d,\eps}^\up(z) \\
\Phi_{\d,\eps}^\rt(z) \\
\Phi_{\d,\eps}^\dn(z) \\
\Phi_{\d,\eps}^\lt(z) 
\end{pmatrix}.
\ee
The \fk--Ising observable~\cite[eq.~(2.4)]{ch-sm} of Chelkak and
Smirnov is given by
\be
\begin{split}
F_{\d,\eps}(z)&=\tfrac12\textstyle{\sum}_\a F_{\d,\eps}^\a(z)\\
&=\cos(\eps/2)[\Phi_\eps^\up(z)+ \Phi_\eps^\dn(z)]
+\cos(\eps^\star/2)[\Phi_\eps^\lt(z)+ \Phi_\eps^\rt(z)],
\end{split}
\ee
where the second line uses~\eqref{mtrx}.
Assuming the limits
\be
\Phi_\d^\up(z)=\lim_{\eps\to 0} \Phi_{\d,\eps}^\up(z),\quad
\Phi_\d^\dn(z)=\lim_{\eps\to 0} \Phi_{\d,\eps}^\dn(z),
\ee
as well as $\Phi_{\d,\eps}^\a(z)=O(\eps)$ for 
$\a\in\{\lt,\rt\}$, we get
\be
\lim_{\eps\to0} F_{\d,\eps}(z)=\Phi_\d^\up(z)+ \Phi_\d^\dn(z),
\ee
which is how we defined our observable $F_\d(z)$.


\subsection{S-holomorphicity}

In this section we show the following result:

\begin{theorem}\label{s-hol-thm-fk}
Let $F_\d=F^{\mathrm{FK}}_\d$ be the \fk-observable in a Dobrushin
domain $(\Dom_\d,a_\d,b_\d)$.  Then $F_\d$ is s-holomorphic in
$\Dom_\d$.
\end{theorem}

It is immediate that $F_\d$ satisfies the
conditions~\eqref{s-hol-def-w-1} and~\eqref{s-hol-def-u-1}
in the definition of s-holomorphicity, see
the discussion just above Definition~\ref{fk-obs-def}.
We thus need to show that also~\eqref{s-hol-def-w-2}
and~\eqref{s-hol-def-u-2} are satisfied.
In the proof we drop the subscript $\d$ from $\EE$ and $\hat\EE$. 

For $z\in\Dom_\d^{\b,\itr}\cup\Dom_\d^{\w,\itr}$ we
let $\xi_z=\xi\sd\{z\}$ and
we define the auxiliary observables 
\be\label{phi-ltrt}
\begin{split}
\Phi_\d^\lt(z)&=\hat\EE[(\sqrt{2})^{L(\xi_z)-L(\xi)}\varphi^\lt(\xi_z;z)],\\
\Phi_\d^\rt(z)&=\hat\EE[(\sqrt{2})^{L(\xi_z)-L(\xi)}\varphi^\rt(\xi_z;z)].
\end{split}
\ee
If $z\in\partial^\rv\Dom_\d$ is in the vertical part of the boundary
then we set $\Phi_\d^\lt(z)=\Phi_\d^\rt(z)=0$, whereas if 
$z\in\partial^\rh\Dom_\d$ is in the horizontal part we define them as
in~\eqref{phi-ltrt} but with $\xi_z$ replaced by $\xi$.
As we remarked above we have that $\Phi_\d^\rt(z)\in\ell(\rt)=\RR$ and
$\Phi_\d^\lt(z)\in\ell(\lt)=i\RR$.
We now claim the following:

\begin{lemma}\label{turn-lem}
For all $z\in\Dom_\d^{\b,\itr}\cup\Dom_\d^{\w,\itr}$ we have that 
\be\label{turn-eq}
\begin{split}
\Phi_\d^\up(z)&=\tfrac{1}{\sqrt2}\big(
e^{i\pi/4}\Phi_\d^\lt(z)+e^{-i\pi/4}\Phi_\d^\rt(z)\big),\\
\Phi_\d^\dn(z)&=\tfrac{1}{\sqrt2}\big(
e^{i\pi/4}\Phi_\d^\rt(z)+e^{-i\pi/4}\Phi_\d^\lt(z)\big).
\end{split}
\ee
\end{lemma}

\begin{proof}
We prove the statement for $\Phi_\d^\up(z)$ in the case when
$z\in\Dom_\d^{\w,\itr}$ is white, the other cases are similar.  
We refer to Figures~\ref{prehol_fig-1},
\ref{prehol_fig-2} and \ref{prehol_fig-3}.  

\begin{figure}[hbt]
\centering
\includegraphics{fermionic.8}\hspace{1cm}
\includegraphics{fermionic.9}
\caption{In $\xi$ the interface $\g$ passes $z$ in direction $\up$ only,
  in $\xi_z$ it passes in directions $\rt$ and $\lt$.}
\label{prehol_fig-1}
\end{figure}

\begin{figure}[hbt]
\centering
\includegraphics{fermionic.10}\hspace{1cm}
\includegraphics{fermionic.11}
\caption{In $\xi$ the interface $\g$ passes $z$ in directions $\up$
  and $\dn$,    in $\xi_z$ it passes in direction  $\lt$.}
\label{prehol_fig-2}
\end{figure}

\begin{figure}[hbt]
\centering
\includegraphics{fermionic.12}\hspace{1cm}
\includegraphics{fermionic.13}
\caption{In $\xi$ the interface $\g$ passes $z$ in directions $\up$
  and $\dn$,    in $\xi_z$ it passes in direction  $\rt$.}
\label{prehol_fig-3}
\end{figure}

Let $A$ denote the event
that $\g$ passes $z$ only once, in the direction $\up$, as depicted on
the left in Figure~\ref{prehol_fig-1}.  Let $A'$ denote the event that
$\g$ passes $z$ once in the direction $\rt$ and once in the direction
$\lt$, with $\rt$ coming first, 
as depicted on the right in Figure~\ref{prehol_fig-1}.
Similarly, let $B$ and $B'$ denote the events depicted in
Figure~\ref{prehol_fig-2}.  Explicitly, $B$ is the event that $\g$
passes $z$ both going $\up$ and $\dn$, with $\dn$ coming first, and
$B'$ is the event that $\g$ passes $z$ in direction $\lt$ only.
Finally, let $C$ and $C'$ be as in Figure~\ref{prehol_fig-3}:
$C$ is the event that $\g$ passes $z$ in direction $\up$ and later in
direction $\dn$, and $C'$ is the event that it passes in direction
$\rt$ only.  

We note the following facts.  Firstly,
\be
\begin{split}
&\xi\in A\Leftrightarrow \xi_z\in A',\mbox{ and then }
L(\xi)=L(\xi_z)+1,\\
&\xi\in B\Leftrightarrow \xi_z\in B',\mbox{ and then }
L(\xi)=L(\xi_z)-1,\\
&\xi\in C\Leftrightarrow \xi_z\in C', \mbox{ and then }
L(\xi)=L(\xi_z)-1.
\end{split}
\ee
Secondly, the event $\G^\up_z=\{\g\mbox{ passes $z$ going }\up\}$
satisfies
\be
\one_{\G^\up_z}(\xi)=\one_A(\xi)+\one_B(\xi)+\one_C(\xi)
\ee
and the events $\G^\lt_z$ and $\G^\rt_z$ satisfy
\be
\begin{split}
\one_{\G^\lt_z}(\xi_z)&=\one_{A'}(\xi_z)+\one_{B'}(\xi_z),
\mbox{ and }\\
\one_{\G^\rt_z}(\xi_z)&=\one_{A'}(\xi_z)+\one_{C'}(\xi_z).
\end{split}
\ee
Thirdly, the winding angles are related by
\be\label{angles-1}
\begin{split}
W^\up_{\g(\xi)}(z)&=W^\lt_{\g(\xi_z)}(z)+\pi/2,
\mbox{ for } \xi\in A\cup B;\\
W^\up_{\g(\xi)}(z)&=W^\rt_{\g(\xi_z)}(z)-\pi/2,
\mbox{ for } \xi\in A\cup C.
\end{split}
\ee
Using these facts, we obtain:
\be
\begin{split}
\varphi^\up(\xi;z)&=(\one_A(\xi)+\one_B(\xi)+\one_C(\xi))
\exp\big(\tfrac i2 W^\up_{\g(\xi)}(z)\big)\\
&=\tfrac12\one_{A'}(\xi_z) 
\exp\big(\tfrac i2 W^\lt_{\g(\xi_z)}(z)\big) e^{i\pi/4}\\
&\quad+\tfrac12\one_{A'}(\xi_z) 
\exp\big(\tfrac i2 W^\rt_{\g(\xi_z)}(z)\big) e^{-i\pi/4}\\
&\quad+\one_{B'}(\xi_z) 
\exp\big(\tfrac i2 W^\lt_{\g(\xi_z)}(z)\big) e^{i\pi/4}\\
&\quad+\one_{C'}(\xi_z) 
\exp\big(\tfrac i2 W^\rt_{\g(\xi_z)}(z)\big) e^{-i\pi/4}.
\end{split}
\ee
Thus
\be
\begin{split}
&(\sqrt2)^{L(\xi)}\varphi^\up(\xi;z)\\
&\quad=(\sqrt2)^{L(\xi)-2}\one_{A'}(\xi_z)
\big\{\exp\big(\tfrac i2 W^\lt_{\g(\xi_z)}(z)\big) e^{i\pi/4}+
\exp\big(\tfrac i2 W^\rt_{\g(\xi_z)}(z)\big) e^{-i\pi/4}\big\}\\
&\quad\quad + (\sqrt2)^{L(\xi)}\one_{B'}(\xi_z)
\exp\big(\tfrac i2 W^\lt_{\g(\xi_z)}(z)\big) e^{i\pi/4}\\
&\quad\quad + 
(\sqrt2)^{L(\xi)}\one_{C'}(\xi_z)
\exp\big(\tfrac i2 W^\rt_{\g(\xi_z)}(z)\big) e^{-i\pi/4}\\
&\quad=(\sqrt2)^{L(\xi_z)-1}
\big[\varphi^\lt(\xi_z;z)e^{i\pi/4}+
\varphi^\rt(\xi_z;z)e^{-i\pi/4}\big].
\end{split}
\ee
Taking the $\EE$-expectation,
\be\begin{split}
&\EE[(\sqrt2)^{L(\xi)}\varphi^\up(\xi;z)]\\
&\quad=\tfrac{1}{\sqrt2}\big(
\EE[(\sqrt2)^{L(\xi_z)}\varphi^\lt(\xi_z;z)]e^{i\pi/4}+
\EE[(\sqrt2)^{L(\xi_z)}\varphi^\rt(\xi_z;z)]e^{-i\pi/4}
\big).
\end{split}\ee
This readily gives the claim~\eqref{turn-eq}
for $\Phi^\up(z)$.
\end{proof}

We now calculate $\dot\Phi_\d^\up$ and $\dot\Phi_\d^\dn$.
We will use 
the notation $\xi(z,z+i\eps)$ for the number of elements of
$\xi$ in the interval $(z,z+i\eps)$.
For a function $f(\xi,z)$ we write
$f(\xi,t\pm)=\lim_{\eps\dn 0}f(\xi,t\pm i\eps)$.
Recall that $\xi_z=\xi\sd\{z\}$.

\begin{lemma}\label{phi-dot-lem}
Let $w\in\Dom_\d^{\w,\itr}$ and write $u=w-\sfrac\d2$
and $v=w+\sfrac\d2$.  Then
\be\label{phi-up-dot}
 \begin{split}
\dot\Phi_\d^\up(w)=\dot\Phi_\d^\up(u)
&=\tfrac1{\d\sqrt2}\big(
e^{i\pi/4}\Phi_\d^\lt(w)-e^{-i\pi/4}\Phi_\d^\rt(w)\big)\\
&\quad +
\tfrac1{\d\sqrt2}\big(
e^{-i\pi/4}\Phi_\d^\rt(u)-e^{i\pi/4}\Phi_\d^\lt(u)\big).
\end{split}
\ee
and
\be\label{phi-dn-dot}
 \begin{split}
\dot\Phi_\d^\dn(w)=\dot\Phi_\d^\dn(v)
&=\tfrac1{\d\sqrt2}\big(
e^{-i\pi/4}\Phi_\d^\lt(w)-e^{i\pi/4}\Phi_\d^\rt(w)\big)\\
&\quad +
\tfrac1{\d\sqrt2}\big(
e^{i\pi/4}\Phi_\d^\rt(v)-e^{-i\pi/4}\Phi_\d^\lt(v)\big).
\end{split}
\ee
\end{lemma}
\begin{proof}
The first equalities in~\eqref{phi-up-dot} and~\eqref{phi-dn-dot}
hold since  $\Phi_\d^\up(w)=\Phi_\d^\up(u)$ and 
$\Phi_\d^\dn(w)=\Phi_\d^\dn(v)$.  We prove~\eqref{phi-up-dot}, the
other claim~\eqref{phi-dn-dot} is similar.  
We have that
\be\label{der-1}
Z_\d\frac{\Phi_\d^\up(w+i\eps)-\Phi_\d^\up(w)}{\eps}
=\tfrac1\eps\EE[(\sqrt2)^{L(\xi)}
(\varphi^\up(\xi;w+i\eps)-\varphi^\up(\xi;w))].
\ee
Note that $\varphi^\up(\xi;w+i\eps)-\varphi^\up(\xi;w)=0$
unless either $\xi(w,w+i\eps)>0$ or $\xi(u,u+i\eps)>0$.
The probability that both these happen is $O(\eps^2)$ and may therefore
be ignored.  Also recall that $\varphi^\up(\xi;w)=\varphi^\up(\xi;u)$
for $w$ and $u$ as specified.  Thus the right-hand-side
of~\eqref{der-1} equals
\be
\begin{split}
&\tfrac1\eps\EE[(\sqrt2)^{L(\xi)}
(\varphi^\up(\xi;w+i\eps)-\varphi^\up(\xi;w))
\one\{\xi(w,w+i\eps)>0\}]\\
&\quad+
\tfrac1\eps\EE[(\sqrt2)^{L(\xi)}
(\varphi^\up(\xi;u+i\eps)-\varphi^\up(\xi;u))
\one\{\xi(u,u+i\eps)>0\}]+o(1).
\end{split}
\ee
This converges to
\be
\begin{split}
&\tfrac1{\d\sqrt2}\EE[(\sqrt2)^{L(\xi_w)}
(\varphi^\up(\xi_w;w+)-\varphi^\up(\xi_w;w-))]\\
&\quad+
\tfrac1{\d\sqrt2}\EE[(\sqrt2)^{L(\xi_u)}
(\varphi^\up(\xi_u;u+)-\varphi^\up(\xi_u;u-))].
\end{split}
\ee

Consider $\varphi^\up(\xi_w;w+)-\varphi^\up(\xi_w;w-)$.
We refer again to Figures~\ref{prehol_fig-1}, \ref{prehol_fig-2} and
\ref{prehol_fig-3} and the events $A,B,C,A',B',C'$ depicted there, as
well as the relation~\eqref{angles-1} between winding angles.
We have that
\be
\begin{split}
\mbox{for }\xi\in A,\quad
&\varphi^\up(\xi_w;w+)-\varphi^\up(\xi_w;w-)=0\\
&\quad=\varphi^\lt(\xi_w;w)e^{i\pi/4}-\varphi^\rt(\xi_w;w)e^{-i\pi/4},
\end{split}
\ee
\be
\begin{split}
\mbox{for }\xi\in B,\quad
&\varphi^\up(\xi_w;w+)-\varphi^\up(\xi_w;w-)=
\varphi^\up(\xi_w;w+)\\
&\quad=\varphi^\lt(\xi_w;w)e^{i\pi/4}\\
&\quad=\varphi^\lt(\xi_w;w)e^{i\pi/4}-\varphi^\rt(\xi_w;w)e^{-i\pi/4},
\end{split}
\ee
\be
\begin{split}
\mbox{for }\xi\in C,\quad
&\varphi^\up(\xi_w;w+)-\varphi^\up(\xi_w;w-)=
-\varphi^\up(\xi_w;w-)\\
&\quad=-\varphi^\rt(\xi_w;w)e^{-i\pi/4}\\
&\quad=\varphi^\lt(\xi_w;w)e^{i\pi/4}-\varphi^\rt(\xi_w;w)e^{-i\pi/4}.
\end{split}
\ee
That is to say, we have the identity
\be
\varphi^\up(\xi_w;w+)-\varphi^\up(\xi_w;w-)=
\varphi^\lt(\xi_w;w)e^{i\pi/4}-\varphi^\rt(\xi_w;w)e^{-i\pi/4}.
\ee
This gives
\be
\EE\big[(\sqrt2)^{L(\xi_w)}
\big(\varphi^\up(\xi_w;w+)-\varphi^\up(\xi_w;w-)\big)\big]=
Z_\d(\Phi_\d^\lt(w)e^{i\pi/4}-\Phi_\d^\rt(w)e^{-i\pi/4}).
\ee
Similar considerations give
\be
\EE\big[(\sqrt2)^{L(\xi_u)}
\big(\varphi^\up(\xi_u;u+)-\varphi^\up(\xi_u;u-)\big)\big]=
Z_\d(\Phi_\d^\rt(w)e^{-i\pi/4}-\Phi_\d^\lt(w)e^{i\pi/4}).
\ee
Combining these and dividing by $Z_\d$ gives the
claim~\eqref{phi-up-dot}.   
\end{proof}

\begin{proof}[Proof of Theorem~\ref{s-hol-thm-fk}]
As already noted,
properties~\eqref{s-hol-def-w-1} and~\eqref{s-hol-def-u-1}
are immediate, so we need to establish~\eqref{s-hol-def-w-2}
and~\eqref{s-hol-def-u-2}.
We check the case $z=w\in\Dom_\d^{\w,\itr}$, the case
$z\in\Dom_\d^{\b,\itr}$ being similar.
Writing $u=w-\sfrac\d2$ 
and $v=w+\sfrac\d2$,  we need to show that
\be\label{phi-dots}
\begin{split}
\dot\Phi_\d^\up(w)&=\dot\Phi_\d^\up(u)=
\tfrac{i}{\d}\big(\Phi_\d^\dn(w)-\Phi_\d^\dn(u)\big),\\
\dot\Phi_\d^\dn(w)&=\dot\Phi_\d^\dn(v)=
\tfrac{i}{\d}\big(\Phi_\d^\up(v)-\Phi_\d^\up(w)\big).
\end{split}
\ee
But for any 
$z\in\Dom_\d^{\w,\itr}\cup\Dom_\d^{\b,\itr}$ 
we have, by Lemma~\ref{turn-lem},  firstly
\be\begin{split}
e^{i\pi/4}\Phi_\d^\lt(z)-e^{-i\pi/4}\Phi_\d^\rt(z)
&=e^{i\pi/2} e^{-i\pi/4}\Phi_\d^\lt(z)-e^{-i\pi/2}e^{i\pi/4}\Phi_\d^\rt(z)\\
&=i\cdot (e^{-i\pi/4}\Phi_\d^\lt(z)+e^{i\pi/4}\Phi_\d^\rt(z))\\
&=i\sqrt2\cdot \Phi_\d^\dn(z),
\end{split}\ee
and secondly
\be\begin{split}
e^{i\pi/4}\Phi_\d^\rt(z)-e^{-i\pi/4}\Phi_\d^\lt(z)
&=e^{i\pi/2} e^{-i\pi/4}\Phi_\d^\rt(z)-e^{-i\pi/2}e^{i\pi/4}\Phi_\d^\lt(z)\\
&=i\cdot (e^{-i\pi/4}\Phi_\d^\rt(z)+e^{i\pi/4}\Phi_\d^\lt(z))\\
&=i\sqrt2\cdot \Phi_\d^\up(z).
\end{split}\ee
Putting these into Lemma~\ref{phi-dot-lem} gives the
result. 
\end{proof}

\section{The spin-observable}
\label{spin-obs-sec}

\subsection{Definition}

Let $\Dom_\d$ be a discrete \emph{dual} domain (see
Section~\ref{dom-sec}).  
We work with the random-parity representation~\eqref{rpr}
in $\Dom_\d^\b$, and as before we set $h=J=\tfrac{1}{2\d}$.  
Recall that the set $\xi=\xi^\w\se\Dom_\d^\w$ of bridges 
is a Poisson process with rate $J$.
Define the `lower boundary' of $\Dom_\d^\b$ as
\[
\partial^-\Dom_\d^\b=\{z\in\partial\Dom_\d^\b:
z-i\eps\not\in\Dom_\d^\b
\mbox{ for all $\eps>0$ small enough}\},
\]
and similarly the `upper boundary' of $\Dom_\d^\b$ as
\[
\partial^+\Dom_\d^\b=\{z\in\partial\Dom_\d^\b:
z+i\eps\not\in\Dom_\d^\b
\mbox{ for all $\eps>0$ small enough}\},
\]
Thus 
$\partial^-\Dom_\d^\b\cup \partial^+\Dom_\d^\b=\partial^\rh\Dom_\d^\b$.

We take two distinct points $a_\d,b_\d$ 
on the boundary $\partial\Dom_\d$ with
$a_\d\in\partial^\rv\Dom_\d^\w\cup\partial^\rh\Dom_\d^\b$ either a
white point on the `sides' or a black point on the `top or bottom',
and $b_\d\in\partial^-\Dom_\d^\b$ on the lower boundary.
In the case when $a_\d\in\partial^\rv\Dom_\d^\w$ we let
$a_\d^\itr=a_\d\pm\sfrac\d2\in\Dom_\d^\b$ 
be the black point in $\Dom_\d$ next to
$a_\d$, so that  $(a_\d,a_\d^\itr)$ is a directed half-edge 
pointing horizontally into $\Dom_\d$,
as in Figure~\ref{spin-dom-fig-1}.
If $a_\d\in\partial^\rh\Dom_\d^\b$ we
let $a_\d^\itr=a_\d$ but sometimes interpret
$a_\d^\itr=a_\d\pm0i$ as a point `just inside' $\Dom_\d^{\b,\itr}$.

\begin{figure}[hbt]
\centering
\includegraphics{fermionic.47}
\hspace{1cm}
\includegraphics{fermionic.48}
\caption{
Dual domain $\Dom_\d$ with $\partial\Dom_\d$ drawn dashed and
$\Dom_\d^\b$ drawn solid.
\emph{Left:} A labelling   $\psi_{a_\d,b_\d}^\xi$, with
  points $u$ satisfying $\psi(u)=1$  marked fat, with blue colour
  for loops and red for the path $\g$.  In this case 
  $a_\d\in\partial^\rv\Dom_\d^\w$.
\emph{Right:} Same domain with a labelling $\psi_{a_\d,z}^\xi$ for 
$z\in\Dom_\d^{\w,\itr}$.   In this case $a_\d\in\partial^+\Dom_\d^\b$.}
\label{spin-dom-fig-1}
\end{figure}

For $a_\d$ as above and for a fixed
$z\in\Dom_\d^\b$, possibly $z=b_\d$, we let
$\psi=\psi_{a_\d,z}^\xi:\Dom_\d^\b\to\{0,1\}$ be 
a function satisfying the following:
\begin{enumerate}
\item $\psi(a_\d^\itr)=\psi(z)=1$ if $z\neq a_\d^\itr$, 
respectively $=0$ if $z=a_\d^\itr$,
\item $\psi(u)=0$ for all $u\in\partial^\rh\Dom^\b_\d\sm\{a_\d,z\}$,
\item for $u\in\Dom_\d^{\b,\itr}$ we have that  
$\psi(u+\eps i)=1-\psi(u-\eps i)$ for all small enough $\eps>0$  
 if \emph{either} $u\pm\sfrac\d2\in\xi$
(that is, $u$ is an endpoint of a bridge) \emph{or}
$u\in(\{a_\d^\itr\}\sd\{z\})$;  and
\item the set $I(\psi)=\{u\in\Dom_\d^\b:\psi(u)=1\}$ is closed. 
\end{enumerate}
Thus $\psi$ is a random-parity configuration with sources
$A=\{a_\d^\itr,z\}$ and boundary condition 0 on
$\partial^\rh\Dom_\d^\b$.
It is easy to see that there is at most one function $\psi_{a_\d,z}^\xi$
satisfying the above constraints, for each given $\xi$ (and $a_\d,z$).
We let $\cA(a_\d,z)$ 
be the event (set of $\xi$:s) such that there
exists such a $\psi$.  We also extend the definition of $\cA(a_\d,z)$ to
allow $z\in\Dom_\d^{\w}$ by letting
\be
\cA(a_\d,z)=\cA(a_\d,z-\sfrac\d2)\cup\cA(a_\d,z+\sfrac\d2)
\mbox{ if } z\in\Dom_\d^{\w}.
\ee
Note that this union is disjoint.

It is worth stating precisely a (necessary and sufficient) condition for
$\xi$ to belong to $\cA({a_\d,z})$ when $z\in\Dom_\d^\b$.
To state the condition, let
\[
V(u)=\{u'\in\Dom_\d^\b:[u,u']\se\Dom_\d^\b\},
\mbox{ for } u\in\Dom_\d^\b,
\]
be the maximal vertical line contained in $\Dom_\d^\b$ and containing
$u$.  Let 
\be
S_{a_\d,z}^\xi(u)=\{v\in V(u): v\pm\sfrac\d2\in\xi\}
   \cup (\{a_\d^\itr\}\sd\{z\})
\ee
be the set of points in $V(u)$ where $\psi$ is required to change value.
Then, for $z\in\Dom_\d^\b$,
\be
\xi\in\cA({a_\d,z})\Leftrightarrow
|S_{a_\d,z}^\xi(u)|\mbox{ is even for all } u\in\Dom_\d^\b.
\ee
In words,
$\psi$ must switch (from 0 to 1 or from 1 to 0)
an even number of times on each line $V(u)$.

If $u\in\Dom_\d^\b\sm\{a_\d^\itr\}$ 
and $\xi\in\cA({a_\d,u})$ then $\psi=\psi^\xi_{a_\d,u}$
contains a unique path $\g(\xi)$ from $a_\d$ to $u$ which traverses 
the half-edge $(a_\d,a_\d^\itr)$ if $a_\d\in\partial^\rv\Dom_\d^\w$, 
intervals along which $\psi=1$, as well as bridges of $\xi$.  
For $w\in\Dom_\d^{\w}$ and $\xi\in\cA({a_\d,w})$
we complete $\g$ to form a path to $w$ by including the half-edge
$(w-\sfrac\d2)\to w$ (if $\xi\in\cA(a_\d,w-\sfrac\d2)$)
respectively
$w\lt(w+\sfrac\d2)$ (if $\xi\in\cA(a_\d,w+\sfrac\d2)$).
See Figure~\ref{spin-dom-fig-1}.
In the cases when $z\in\{a_\d,a_\d^\itr\}$ the path $\g$ is
degenerate, and we interpret it as a small arrow (or half-edge)
pointing from $a_\d$ to $a_\d^\itr$ if $z=a_\d^\itr$, alternatively as
a small path making an angle $\pi$ turn if $z=a_\d$.

We define $W^{a_\d,z}_{\g(\xi)}$ to be the
winding-angle of $\g(\xi)$ from $a_\d$ to $z$
(with $W^{a_\d,a_\d^\itr}=0$ and $W^{a_\d,a_\d}=\pi$).  It is important to note
that, in the case when $z=b_\d$ is on the boundary, then 
$W^{a_\d,b_\d}_{\g(\xi)}$ does not depend on $\xi$ (one cannot wind around the
boundary, and $a_\d,b_\d$ have fixed orientations), i.e.\ it takes a fixed
value which we denote $W^{a_\d,b_\d}$.
 
Write $1^\w(z)$ for the indicator that $\g$ ends with a half-edge
(i.e.\ either $z\in\Dom_\d^\w$ or $a_\d\in\partial^\rv\Dom_\d^\w$
and $z=a_\d^\itr$).
Define the random variable
\be
X^{a_\d,z}(\xi)=\one_{\cA(a_\d,z)}(\xi)
\exp(-2h|I(\psi_{a_\d,z}^\xi)|) 
(\tfrac1{\sqrt2})^{1^\w(z)}.
\ee

\begin{definition}\label{spin-obs-def}
Write $\EE=\EE_{0,1/2\d}$ for the law of $\xi=\xi^\w$ and let
$(\Dom_\d,a_\d,b_\d)$ be as above.
Define the spin-observable
\be
F^{\mathrm{sp}}_\d(z)=
\frac{\EE[\exp(-\tfrac{i}2W^{a_\d,z}_{\g(\xi)}) X_{\g(\xi)}^{a_\d,z}]}
{\EE[\exp(-\tfrac{i}2W^{a_\d,b_\d}_{\g(\xi)}) X_{\g(\xi)}^{a_\d,b_\d}]},
\quad z\in\Dom_\d^\b\cup\Dom_\d^{\w}.
\ee
\end{definition}

Note that we have defined this observable using the random-parity
representation, whose classical analogue is the random-current
representation of~\cite{aiz82} rather than the high-temperature expansion
used by Chelkak and Smirnov~\cite{ch-sm}. 
The high-temperature expansion is essentially the random-current
representation `modulo two'.

\subsection{S-holomorphicity}

In this section we show the following result.

\begin{theorem}\label{s-hol-thm-spin}
Let $F^{\mathrm{sp}}_\d$ be the spin-observable in a primal domain 
$(\Dom_\d,a_\d,b_\d)$ with two marked points on the
boundary, as above.  Then $F^{\mathrm{sp}}_\d$ is s-holomorphic at
all $z\in(\Dom_\d^{\b,\itr}\cup\Dom_\d^{\w,\itr})\sm\{a_\d^\itr,a_\d^\itr\pm\sfrac\d2\}$.
\end{theorem}

Regarding the behaviour near $a_\d^\itr$, we note that half of
condition~\eqref{s-hol-def-u-1} in Definition~\ref{s-hol-def} 
holds at $a_\d^\itr$, but condition~\eqref{s-hol-def-u-2} fails.

Since $a_\d$ and $b_\d$ are fixed 
we will use the shorthands
\be
\cW^z(\xi)=W^{a_\d,b_\d}-W^{a_\d,z}_{\g(\xi)},
\qquad
X^z(\xi)=X^{a_\d,z}(\xi).
\ee
Note that $F^{\mathrm{sp}}_\d(z)$ is a real multiple of
\be\label{sp-obs-pf-eq}
F_\d(z)=\EE\big[\exp(\tfrac{i}2 \cW^z(\xi)) 
X^{z}(\xi)\big],
\ee
so it suffices to show s-holomorphicity of this $F_\d(z)$.

It will be useful to note the
following interpretation of the quantity $\cW^{z}(\xi)$.  
Imagine that we augment $\g$ with a curve $\hat\g$ in $\Dom_\d$ 
which starts at $z$
in the same direction that $\g$ ends, and which finishes at $b_\d$
(pointing down).  Let $W^{z,b_\d}_{\hat\g}$ denote its winding angle.
Then $\G=\g\cup\hat\g$ is a curve in
$\Dom_\d^\b$ from $a_\d$ to $b_\d$, 
thus $\G$ has winding angle $W^{a_\d,b_\d}$,
meaning that $\cW^{z}(\xi)$ is the winding-angle from $z$ to $b_\d$. 

We now turn to the proof of Theorem~\ref{s-hol-thm-spin}.   
Recall that we define $F_\d^\a(z)$ by
\[
F_\d^\a(z)=\mathrm{Proj}[ F_\d(z); \ell(\a)],\quad
\a\in\{\up,\dn,\lt,\rt\}.
\]
This means that the $F_\d^\a$
automatically satisfy the relations of the $\Phi_\d^\a$ 
in  Lemma~\ref{turn-lem}, that is:
\be\label{turn-eq-2}
\begin{split}
F_\d^\up(z)&=\tfrac{1}{\sqrt2}\big(
e^{i\pi/4}F_\d^\lt(z)+e^{-i\pi/4}F_\d^\rt(z)\big),\\
F_\d^\dn(z)&=\tfrac{1}{\sqrt2}\big(
e^{i\pi/4}F_\d^\rt(z)+e^{-i\pi/4}F_\d^\lt(z)\big).
\end{split}
\ee
If $z=u\in\Dom_\d^\b\sm\{a_\d^\itr\}$  
then $\g$ reaches $u$ either
from below or from above;  we write these events pictorially as
\[
\Big\{\upu\Big\}
\qquad\mbox{and}\qquad
\Big\{\dnu\Big\}.
\]
Similarly, if $z=w\in\Dom_\d^\w$ then $\g$ reaches $w$ either from the
left or the right, pictorially represented as
\[
\big\{\ltw\big\} \qquad\mbox{and}\qquad
\big\{\rtw\big\}.
\]

\begin{lemma}
If  $u\in\Dom_\d^\b\sm\{a_\d^\itr\}$ then
\be\label{Phi-u-proj}
\begin{split}
F_\d^\lt(u)&=\EE\Big[
\exp\big(\tfrac{i}2\cW^{u}\big) X^{u}
\one\Big\{\upu\Big\}
\Big],\mbox{ and}\\
F_\d^\rt(u)&=\EE\Big[
\exp\big(\tfrac{i}2\cW^{u}\big) X^{u}
\one \Big\{\dnu\Big\}
\Big],
\end{split}
\ee
and if  $w\in\Dom_\d^\w$,
\be\label{Phi-w-proj}
\begin{split}
F_\d^\up(w)&=\EE\Big[
\exp\big(\tfrac{i}2\cW^{w}\big) X^{w}
\one\big\{\rtw\big\}
\Big],\mbox{ and}\\
F_\d^\dn(w)&=\EE\Big[
\exp\big(\tfrac{i}2\cW^{w}\big) X^{w}
\one\big\{\ltw\big\}
\Big].
\end{split}
\ee
\end{lemma}
\begin{proof}
We show~\eqref{Phi-u-proj}, the argument for~\eqref{Phi-w-proj} is
similar. 
Certainly the two terms on the right-hand-sides of~\eqref{Phi-u-proj}
sum to $F_\d(u)$.  Moreover, if $\g$ reaches $u$ from below
then $\cW^{u}=\pi+2\pi n$ for some $n=n(\xi)\in\ZZ$, and if 
$\g$ reaches $u$ from above
then $W^{u}=0+2\pi n$ for some $n=n(\xi)\in\ZZ$.  
Thus the two terms belong to
$\ell(\lt)$ and $\ell(\rt)$ respectively, and these two lines being
perpendicular,  the claim~\eqref{Phi-u-proj} follows.  
\end{proof}

\begin{proposition}
Conditions~\eqref{s-hol-def-w-1} and~\eqref{s-hol-def-u-1} in
Definition~\ref{s-hol-def} hold at all
$z\in(\Dom_\d^{\b,\itr}\cup\Dom_\d^{\w,\itr})\sm\{a_\d^\itr,a_\d^\itr\pm\sfrac\d2\}$. 
\end{proposition}
\begin{proof}
We give details for the case  $z=w\in\Dom_\d^{\w,\itr}$, the case
$z\in\Dom_\d^{\b,\itr}$ being similar. 
Writing $u=w-\sfrac\d2$, $v=w+\sfrac\d2$, we need to show that
(when neither $u$ nor $v$ equals $a_\d^\itr$)
\[ 
F_\d^\up(w)=F_\d^\up(u) \mbox{ and }
F_\d^\dn(w) =F_\d^\dn(v).
\]
We give details for the case $\up$ only, the claim for $\dn$ again
being similar. 

Consider the terms in~\eqref{Phi-u-proj}.
Inside the expectations we have
\be \label{Phi-u-wind-1}
\mbox{in }F_\d^\lt(u),\quad
X^{u}(\xi)={\sqrt2} X^{w}(\xi)
\mbox{ and }
W^{a_\d,u}_{\g(\xi)}=W^{a_\d,w}_{\g(\xi)}+\tfrac\pi2,
\ee
since if we add the half-edge from $u$ to $w$,
this puts an additional factor 
$\sfrac1{\sqrt2}$ into $X$, and $\g$ does an additional 
$-\sfrac\pi2$ turn.  Similarly,
\be \label{Phi-u-wind-2}
\mbox{in }F_\d^\rt(u),\quad
X^{u}(\xi)={\sqrt2} X^{w}(\xi)
\mbox{ and }
W^{a_\d,u}_{\g(\xi)}=W^{a_\d,w}_{\g(\xi)}-\tfrac\pi2.
\ee
We use the symbolic notation
\[
\Big\{\uprtw\Big\} \qquad\mbox{and}\qquad
\Big\{\dnrtw\Big\}
\]
for the events that $\g$ ends with a right- or left-turn at $u$ into
$w$, respectively.  Using~\eqref{turn-eq-2}, \eqref{Phi-u-proj},
\eqref{Phi-u-wind-1} and \eqref{Phi-u-wind-2}, we have for the case
when neither $u$ nor $v$ equals $a_\d^\itr$:
\be
\begin{split}
F_\d^\up(u)&=\tfrac{1}{\sqrt2}\big(
e^{i\pi/4}F_\d^\lt(u)+e^{-i\pi/4}F_\d^\rt(u)\big)\\
&=\tfrac1{\sqrt2}\EE\Big[
\exp\big(\tfrac{i}2(\cW^{u}+\tfrac\pi2)\big)
X^{u} \one\Big\{\upu\Big\}
\Big]\\
&\qquad+\tfrac1{\sqrt2}\EE\Big[
\exp\big(\tfrac{i}2(\cW^{u}-\tfrac\pi2)\big)
X^{u} \one\Big\{\dnu\Big\}
\Big]\\
&=\EE\Big[
\exp\big(\tfrac{i}2\cW^{w}\big) X^{w} 
\one\Big\{\uprtw\Big\}
\Big]  \\ &\qquad
+\EE\Big[
\exp\big(\tfrac{i}2\cW^{w}\big) X^{w}
\one\Big\{\dnrtw\Big\}
\Big]\\
&=\EE\Big[
\exp\big(\tfrac{i}2\cW^{w}\big) X^{w}
\one\big\{\rtw\big\}
\Big]\\
&=F_\d^\up(w),\mbox{ as required.}\qedhere
\end{split}
\ee
\end{proof}

The remaining conditions for s-holomorphicity take more work to
verify. 
Theorem~\ref{s-hol-thm-spin} follows once we establish the following:

\begin{proposition}
Conditions~\eqref{s-hol-def-w-2} and~\eqref{s-hol-def-u-2} in
Definition~\ref{s-hol-def} hold at all
$z\in(\Dom_\d^{\b,\itr}\cup\Dom_\d^{\w,\itr})\sm\{a_\d^\itr,a_\d^\itr\pm\sfrac\d2\}$.
\end{proposition}
\begin{proof}
Again we give details only for $z=w\in\Dom_\d^{\w,\itr}$.  Writing
$u=w-\sfrac\d2$, $v=w+\sfrac\d2$, we need to show
(as long as neither $u$ nor $v$ equals $a_\d^\itr$) that
\[
\begin{split}
\dot F_\d^\up(w)=\tfrac i\d 
\big(F_\d^\dn(v)-F_\d^\dn(u)\big)\mbox{ and }
\dot F_\d^\dn(w)=\tfrac i\d 
\big(F_\d^\up(v)-F_\d^\up(u)\big).
\end{split}
\]
We give details only for the case of $\dot F_\d^\up(w)$.
Take $\eps>0$ small, and consider 
$F_\d^\up(w+i\eps)-F_\d^\up(w)$.  Note from~\eqref{Phi-w-proj} that 
\be
F_\d^\up(w)=\tfrac1{\sqrt2}\EE\Big[
\exp\big(\tfrac{i}2\cW^w\big) 
\exp\big(-2h|I(\psi_{a_\d,u})|\big)\one_{\cA(a_\d,u)}
\Big].
\ee
Also note that $\cA(a_\d,u)=\cA(a_\d,u+i\eps)$ for $\eps>0$ small.
We may thus write
\be\label{phi-dot-exp-1}
\begin{split}
&F_\d^\up(w+i\eps)-F_\d^\up(w)\\
&=\tfrac1{\sqrt2} \EE\Big[ 
\Big(e^{\sfrac{i}2\cW^{w+i\eps}}
e^{-2h|I(\psi_{a_\d,u+i\eps})|} - e^{\sfrac{i}2\cW^{w}}
e^{-2h|I(\psi_{a_\d,u})|}\Big) \one_{\cA(a_\d,u)}
\Big].
\end{split}
\ee
We will split the expectation into the two cases:
(i) $\xi(w,w+i\eps)=0$, and
(ii) $\xi(w,w+i\eps)>0$, i.e.\ according to whether there is
a bridge in the interval $(w,w+i\eps)$ or not.

\begin{figure}[hbt]
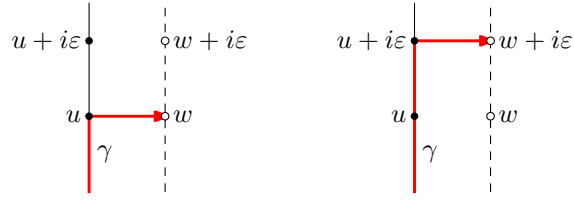

\centering
\includegraphics{fermionic.16}\hspace{1cm}
\includegraphics{fermionic.17}
\caption{The curve $\g$ finishes with a right turn at $u$ (in
  $\psi^\xi_{a,w}$, displayed left)
  respectively $u+i\eps$ (in $\psi^\xi_{a_\d,w+i\eps}$, displayed right).  
The winding angle is the same in  both cases.}
\label{spin-der-fig-1}
\end{figure}
\begin{figure}[hbt]
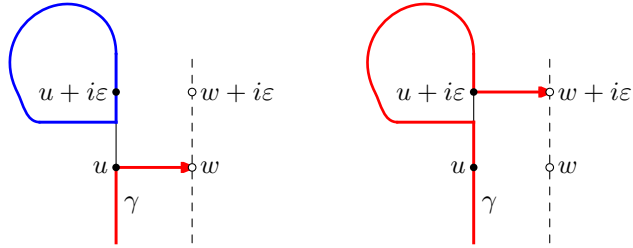

\centering
\includegraphics{fermionic.18}\hspace{1cm}
\includegraphics{fermionic.19}
\caption{Here $u+i\eps$ is contained in a loop  (in
  $\psi^\xi_{a_\d,w}$, displayed left), which becomes part of
  $\g$ (in $\psi^\xi_{a_\d,w+i\eps}$, displayed right).  
 The winding angle is still the same in
  both cases.}
\label{spin-der-fig-2}
\end{figure}

The first case, when there is no bridge, 
is illustrated in Figures~\ref{spin-der-fig-1}
and~\ref{spin-der-fig-2}.
In this case we have that
$W^{a_\d,w+i\eps}_\g=W^{a_\d,w}_\g$
and hence $\cW^w=\cW^{w+i\eps}$.
 Let
\be
\hat\eps=|I(\psi_{a_\d,u})|-|I(\psi_{a_\d,u+i\eps})|,
\ee
and note that $-\eps\leq\hat\eps\leq\eps$.
We may thus write the contribution from case (i) to the 
expectation in~\eqref{phi-dot-exp-1} as 
\be\label{phi-dot-exp-2} 
\tfrac1{\sqrt2}\EE\Big[
e^{\sfrac{i}2\cW^w}
e^{-2h|I(\psi_{a_\d,u})|}
\big(e^{2h\hat\eps}-1\big) \one_{\cA(a_\d,u)}
\one\{\xi(w,w+i\eps)=0\}
\Big].
\ee
Since the factor $e^{2h\hat\eps}-1$ is of order $O(\eps)$ we can (up
to an error of order $O(\eps^2)$) ignore events of probability
$O(\eps)$.  Thus we may assume that there is no bridge
in $(w-\d,w-\d+i\eps)$ (i.e.\ we have a situation as in
Figure~\ref{spin-der-fig-1}, not as in
Figure~\ref{spin-der-fig-2}).  Under the latter assumption we have
that  
\be
\hat\eps=\left\{\begin{array}{ll}
+\eps, & \mbox{if $\g$ comes from above},\\
-\eps, & \mbox{if $\g$ comes from below}.
\end{array}\right.
\ee
Thus, up to an error of order $O(\eps^2)$, the integrand
in~\eqref{phi-dot-exp-2} equals
\be\begin{split}
2h\eps\Big( e^{\sfrac{i}2\cW^w}
e^{-2h|I(\psi_{a_\d,u})|}\one\{\dnrtw\}
- e^{\sfrac{i}2\cW^w}
e^{-2h|I(\psi_{a_\d,u})|}\one\{\uprtw\}\Big).
\end{split}\ee 
In the first term we have $W_\g^{a_\d,w}=W_\g^{a_\d,u}+\pi/2$
and in the second  term we have $W_\g^{a_\d,w}=W_\g^{a_\d,u}-\pi/2$.
Dividing by $\eps$ and letting $\eps\downarrow0$,
it follows that the contribution to $\dot F_\d^\up(w)$ from case
(i) is
\be
\begin{split}
& h\sqrt2 (e^{-i\pi/4}F_\d^\rt(u) -e^{i\pi/4}F_\d^\lt(u) )\\
&\quad=-ih\sqrt2 (e^{i\pi/4}F_\d^\rt(u) +e^{-i\pi/4}F_\d^\lt(u) )\\
&\quad=-2ih F_\d^\dn(u)= -\tfrac{i}{\d}F_\d^\dn(u).
\end{split}
\ee

We now turn to case (ii), when there is a bridge in $(w,w+i\eps)$.
We need to show that the contribution
from this case is $\tfrac i\d F_\d^\dn(w)=\tfrac i\d F_\d^\dn(v)$.
We start by noting that, up to an error of order $O(\eps^2)$, we may
in fact assume that $\xi$ belongs to the event
\be\label{B}
B=\left\{\begin{array}{l}
\xi(w,w+i\eps)=1,\\
\xi(w-\d,w-\d+i\eps)=0,\mbox{ and}\\
\xi(w+\d,w+\d+i\eps)=0.
\end{array}\right\}
\ee
The possible scenarios are illustrated in Figures~\ref{case-iia-fig},
\ref{case-iib-fig}, \ref{case-iic-fig} and \ref{case-iid-fig}.
We write $\hat w$ for the location of the unique bridge in
$(w,w+i\eps)$.  Recall the notation 
$\xi_{\hat w}=\xi\sd\{\hat w\}$ for the
configuration obtained by removing
 the bridge at $\hat w$ from $\xi$.
We have that 
\be
\xi\in\cA(a_\d,u) \Leftrightarrow \xi\in\cA(a_\d,u+i\eps) 
\Leftrightarrow \xi_{\hat w}\in\cA(a_\d,v).
\ee
Moreover, we have that the quantities
\be
\begin{split}
\hat\eps_1&=|I(\psi_{a_\d,v}^{\xi_{\hat w}})| 
                         -|I(\psi_{a_\d,u+i\eps}^{\xi})|\\
\hat\eps_2&=|I(\psi_{a_\d,v}^{\xi_{\hat w}})|
                         -|I(\psi_{a_\d,u}^{\xi})|
\end{split}
\ee
satisfy $-2\eps\leq \hat\eps_1,\hat\eps_2\leq 2\eps$.
The contribution from case (ii) to the expectation
in~\eqref{phi-dot-exp-1} may thus, up to an error of order
$O(\eps^2)$,  be written as
\be\label{phi-dot-exp-3}
\begin{split}
&\EE\Big[\one_B(\xi) \one_{\cA(a_\d,v)}(\xi_{\hat w})
X^{w}(\xi_{\hat w})
\Big(e^{\sfrac{i}2\cW^{w+i\eps}(\xi)}
e^{2h\hat \eps_1} - 
e^{\sfrac{i}2\cW^{w}(\xi)}
e^{2h\hat \eps_2} \Big) \Big]\\
&=\EE\Big[\one_B(\xi) \one_{\cA(a_\d,v)}(\xi_{\hat w})
X^{w}(\xi_{\hat w})
\Big(e^{\sfrac{i}2\cW^{w+i\eps}(\xi)}
 - e^{\sfrac{i}2\cW^{w}(\xi)}
\Big) \Big]+O(\eps^2).
\end{split}
\ee
We used that the event $B$ has probability $O(\eps)$ and that both 
$e^{2h\hat\eps_1}=1+O(\eps)$ and $e^{2h\hat\eps_2}=1+O(\eps)$.

It remains to understand the factor
\[
\one_B(\xi) \one_{\cA(a_\d,v)}(\xi_{\hat w})
\big(
e^{\sfrac{i}2\cW^{w+i\eps}(\xi)}
 - e^{\sfrac{i}2\cW^{w}(\xi)}
\big)
\]
We claim that, for  $\xi\in B$ and $\xi_{\hat w}\in\cA(a_\d,v)$,
\be\label{phi-dot-exp-4}
e^{\sfrac{i}2\cW^{w+i\eps}(\xi)}
 - e^{\sfrac{i}2\cW^{w}(\xi)}
=2i\cdot 
e^{\sfrac{i}2\cW^{w}(\xi_{\hat w})}.
\ee
Before showing this, we explain how to finish the proof.  
From~\eqref{phi-dot-exp-3}, and assuming~\eqref{phi-dot-exp-4}, the
contribution to $\dot F_\d^\up(w)$ from case (ii) is
\be\label{phi-dot-exp-5}
\begin{split}
&2i\cdot \lim_{\eps\dn0} \tfrac1\eps
\EE\Big[\one_B(\xi) \one_{\cA(a_\d,v)}(\xi_{\hat w})
X^{w}(\xi_{\hat w})
e^{\sfrac{i}2\cW^{w}(\xi_{\hat w})}
\Big]\\
&=
2i J \EE\Big[
e^{\sfrac{i}2\cW^{w}(\xi)}
X^{w}(\xi) \one_{\cA(a_\d,v)}(\xi)
\Big]\\
&=\tfrac i\d F_\d^\dn(w),
\end{split}
\ee
as required (we used~\eqref{Phi-w-proj}).

It remains to show~\eqref{phi-dot-exp-4}.
There are 4 sub-cases to consider, depending on whether $\g$ traverses
$\hat w$ (in $\psi_{a_\d,u}^\xi$), in which direction, et.c.  
The first case, which we call case (a), is defined by the condition
$\psi_{a_\d,u}^\xi(u+0i)=1$ and is depicted in Figure~\ref{case-iia-fig}.
In this case $\g(\xi)$ necessarily traverses $\hat w$ from right to
left.  

\begin{figure}[hbt]
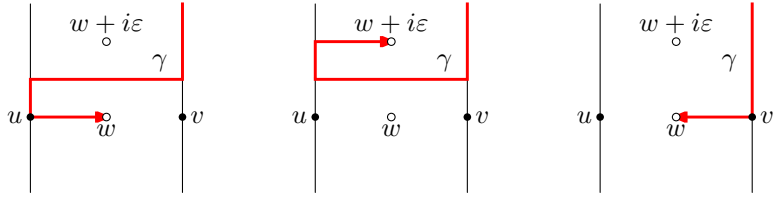

\centering
\includegraphics{fermionic.20}\hspace{1cm}
\includegraphics{fermionic.21}\hspace{1cm}
\includegraphics{fermionic.22}
\caption{Case (ii)(a), with $\psi^\xi_{a_\d,w}$ to the left,
  $\psi^\xi_{a_\d,w+i\eps}$ in the middle, and 
$\psi^{\xi_{\hat w}}_{a_\d,w}$ to the right.} 
\label{case-iia-fig}
\end{figure} 
It is not hard to see that we get
\be
\mbox{case (a):}\quad
W^{a_\d,w}_{\g(\xi)}= W^{a_\d,w}_{\g(\xi_{\hat w})}+\pi,\quad
W^{a_\d,w+i\eps}_{\g(\xi)}= W^{a_\d,w}_{\g(\xi_{\hat w})}-\pi.
\ee 
This establishes~\eqref{phi-dot-exp-4} for case (a).
In the remaining 3 cases we will have $\psi^\xi_{a_\d,u}(u+0i)=0$,
meaning that (in $\psi^\xi_{a_\d,u}$) $\g$ can traverse $\hat w$
from right to left (case (b)), from left to right (case (c)), or not
at all (case (d)).  The cases are depicted in
Figures~\ref{case-iib-fig}, 
\ref{case-iic-fig} and \ref{case-iid-fig}, 
respectively.  We get the following:
\be
\mbox{case (b):}\quad
W^{a_\d,w}_{\g(\xi)}= W^{a_\d,w}_{\g(\xi_{\hat w})}+\pi,\quad
W^{a_\d,w+i\eps}_{\g(\xi)}= W^{a_\d,w}_{\g(\xi_{\hat w})}-\pi.
\ee 
\be
\mbox{case (c):}\quad
W^{a_\d,w}_{\g(\xi)}= W^{a_\d,w}_{\g(\xi_{\hat w})}-3\pi,\quad
W^{a_\d,w+i\eps}_{\g(\xi)}= W^{a_\d,w}_{\g(\xi_{\hat w})}-\pi.
\ee 
\be
\mbox{case (d):}\quad
W^{a_\d,w}_{\g(\xi)}= W^{a_\d,w}_{\g(\xi_{\hat w})}+\pi,\quad
W^{a_\d,w+i\eps}_{\g(\xi)}= W^{a_\d,w}_{\g(\xi_{\hat w})}+3\pi.
\ee 
In all cases we see that~\eqref{phi-dot-exp-4} holds, as claimed.
\begin{figure}[hbt]
\centering
\includegraphics{fermionic.23}\hspace{1cm}
\includegraphics{fermionic.24}\hspace{1cm}
\includegraphics{fermionic.25}
\caption{Case (ii)(b), with $\psi^\xi_{a_\d,w}$ to the left,
  $\psi^\xi_{a_\d,w+i\eps}$ in the middle, and 
$\psi^{\xi_{\hat w}}_{a_\d,w}$ to the right.} 
\label{case-iib-fig}
\end{figure} 
\begin{figure}[hbt]
\centering
\includegraphics{fermionic.26}\hspace{1cm}
\includegraphics{fermionic.27}\hspace{1cm}
\includegraphics{fermionic.28}
\caption{Case (ii)(c), with $\psi^\xi_{a_\d,w}$ to the left,
  $\psi^\xi_{a_\d,w+i\eps}$ in the middle, and 
$\psi^{\xi_{\hat w}}_{a_\d,w}$ to the right.} 
\label{case-iic-fig}
\end{figure} 
\begin{figure}[hbt]
\centering
\includegraphics{fermionic.29}\hspace{1cm}
\includegraphics{fermionic.30}\hspace{1cm}
\includegraphics{fermionic.31}
\caption{Case (ii)(d), with $\psi^\xi_{a_\d,w}$ to the left,
  $\psi^\xi_{a_\d,w+i\eps}$ in the middle, and 
$\psi^{\xi_{\hat w}}_{a_\d,w}$ to the right.} 
\label{case-iid-fig}
\end{figure} 
\end{proof}

\newpage
\section{Discussion}\label{disc-sec}

\subsection{Convergence of the observables}

As mentioned in the Introduction we expect that
both the \fk- and
spin-observables, suitably rescaled, converge as $\d\to0$.  
We sketch an outline of a possible argument,
following the arguments for the classical case 
(see~\cite{ch-sm,dum-cop,smi-fk}).
As also mentioned, the details in the case of the
\fk-observable were supplied by Li~\cite{li}
shortly after this paper was finished.
We take the discrete domains $(\Dom_\d,a_\d,b_\d)$ to approximate a
continuous 
domain $(\Dom,a,b)$ (e.g.\ in the Carath\'eodory sense, i.e.\  convergence on
compact subsets of suitably normalized conformal maps from the upper 
half-plane into the domains, see~\cite[Definition~3.10]{dum-cop}). 

The two main steps are to show (i) precompactness of
sequences of s-holomorphic
functions $(F_\d)_{\d>0}$, and (ii) convergence of the auxiliary
functions $(H_\d)_{\d>0}$ given in Proposition~\ref{H-prop}.

For (i), note first that
preholomorphic functions, and hence in particular
s-holomorphic functions, are $\D_\d$-harmonic.  Indeed,
if $F_\d$ satisfies~\eqref{prehol} at $z$ and $z\pm\sfrac\d2$, then
differentiating twice using~\eqref{prehol} gives
\be\begin{split}
\ddot F_\d(z)&=\tfrac1{i\d}\big(
\dot F_\d(z-\sfrac\d2)-\dot F_\d(z+\sfrac\d2) \big)\\
&=-\tfrac1{\d^2}
\big(F_\d(z-\d)+F_\d(z+\d)-2F_\d(z)\big).
\end{split}\ee
Thus precompactness of s-holomorphic functions would follow from
Lipschitzness of $\D_\d$-harmonic functions combined with a suitable
boundedness condition, using the Arzela--Ascoli theorem as
in~\cite[Proposition~8.7]{dum-cop}.  
Completing this argument would require estimates
for the Green's function $G_\d(\cdot)$ in $\CC_\d^\b$, 
in particular a suitable form of the asymptotics of
$G_\d(z)$ as $|z|\to\oo$ as in~\cite{kenyon} and~\cite{bucking,ch-sm-cplx}.
See Section~3.4 of Li's paper~\cite{li} for details in the
present context.

For (ii), consider the sub- and superharmonic functions 
$H^\b_\d=H_\d\!\mid_{\Dom_\d^\b}$ and
$H^\w_\d=H_\d\!\mid_{\Dom_\d^\w}$ (see Proposition~\ref{H-prop}).  It is
not hard to partly determine the behaviour of these functions on the
boundary.  In the case of the \fk-observable we can choose the
additive constant so that  $H^\b_\d=1$
on the black part $\partial^\b_\d$ and $H_\d^\w=0$ on the white part
$\partial^\w_\d$.  
In the case of the spin-observable 
the constant  can be chosen so
that $H_\d^\w(w)=0$ for all $w\in\partial\Dom_\d^\w\sm\{a_\d\}$
(note also that $\nu(z)^{1/2} F_\d^\mathrm{sp}(z)\in\RR$ for all
$z\in\partial^\rv\Dom_\d^\w\cup\partial^\rh\Dom_\d^\b$ 
where $\nu(z)$ is the counter-clockwise
oriented unit tangent).

To fully determine the boundary-behaviour one could try to use a
variant of the `boundary modification trick' of~\cite{ch-sm}
(this is the approach taken by Li~\cite{li}).
 In the case of
the \fk-observable one could alternatively note that the difference of
$H_\d$ on the boundary 
and `just inside' the boundary is proportional to a
percolation-probability which converges to zero away from $a_\d,b_\d$,
like in the original argument for the square-lattice case~\cite{smi-fk}
(this uses that the phase-transition is continuous~\cite{bjogr}).  
Having determined the boundary-values of $H^\b_\d$ and $H^\w_\d$ one
would show that these functions are close to the harmonic function $h$
in $\Dom$ 
with the corresponding boundary-values.
In the case of the \fk-observable we have $h=1$ on the clockwise arc
from $a$ to $b$ and $h=0$ on the counter-clockwise arc, whereas for
the spin-observable we have $h=0$ on $\partial\Dom\sm\{a\}$.

Since these are  also the boundary-conditions for the
classical case~\cite{ch-sm,smi-fk}, we expect the 
observables to converge to the same limits under the same
rescaling, namely
\be
\tfrac1{\sqrt{\d}} F_\d^\mathrm{FK}(\cdot)\to \sqrt{\phi'(\cdot)},
\quad
F_\d^\mathrm{sp}(\cdot)\to\sqrt{\tfrac{\psi'(\cdot)}{\psi'(b)}},
\ee
where $\phi$ is a conformal map from $\Dom$ to
$\RR+i(0,1)$ mapping $a$ to $-\oo$ and $b$ to $+\oo$,
and $\psi$ is a conformal map from $\Dom$ to  the upper half-plane
mapping $a$ to $\oo$ and $b$ to 0.
As mentioned, the first of these limits has now
been established by Li~\cite{li}.

\subsection{Parafermionic observables}

Recall from~\eqref{fk-dens} that the \fk--Ising model at the critical
parameters $h=J=1/2\d$ has density proportional to $(\sqrt2)^{L(\xi)}$
with respect to a Poisson law, where $L(\xi)$ is
the number of loops.  It is natural to ask also about measures with
density $(\sqrt q)^{L(\xi)}$ for other $q>0$.  
Such measures arise in
the Aizenman--Nachtergaele representation~\cite{aizenman_nacht} 
of a class of quantum spin
systems which includes the (spin-$\tfrac12$) Heisenberg
antiferromagnet as the case $q=4$.  One may define an analog of the
\fk--Ising observable~\eqref{fk-obs-eq} 
which is also a direct analog of Smirnov's parafermionic
observable for critical random-cluster models~\cite{smi-fk}.  We
briefly describe this now.

Let $(\Dom_\d,a_\d,b_\d)$ be a Dobrushin-domain as in
Section~\ref{fk-obs-sec} and let $\s$ satisfy
$\sin(\s\tfrac\pi2)=\tfrac12\sqrt q$.
Thus $\s=\tfrac12$ for $q=2$ (\tfim) and $\s=1$ for $q=4$ (Heisenberg
model).  
Recall the events
$\G^\a_z=\{\g\mbox{ passes $z$ in direction }\a\}$
and the winding-angle $W^\a_\g(z)$ of the interface 
to the exit.  We now define
\be
\varphi^\a(\xi; z)=\one_{\G^\a_z}(\xi)\exp(i\s W^\a_{\g(\xi)}(z)).
\ee
Let $\hat\EE_\d$ denote the measure with density proportional to 
$(\sqrt q)^{L(\xi)}$ with respect to the Poisson law with rate 
$\tfrac1{\d\sqrt q}$.  
Similarly to before we define observables
\[
\Phi_\d^\up(z)=\hat\EE_\d[\varphi^\up(\xi;z)],\quad
\Phi_\d^\dn(z)=\hat\EE_\d[\varphi^\dn(\xi;z)],
\]
as well as
$F_\d(z)=\Phi_\d^\up(z)+\Phi_\d^\dn(z)$.
Some properties of these quantities are immediate, e.g.\ for
$w\in\Dom_\d^{\w,\itr}$ we still have
$\Phi^\up_\d(w)=\Phi_\d^\up(w-\sfrac\d2)$ and
$\Phi^\dn_\d(w)=\Phi_\d^\dn(w+\sfrac\d2)$, and also a version of
Lemma~\ref{phi-dot-lem} holds.  It might be interesting to investigate
these observables further, especially due to the connection with the
Heisenberg antiferromagnet.

\subsection*{Acknowledgement}
This work was mainly carried out while the author
was at the University of Copenhagen in Denmark.
The author is now supported by
Vetenskapsr{\aa}det grant 2015-05195.


\end{document}